\documentclass[12pt,letterpaper]{article}

\usepackage{natbib,hyperref}

\usepackage[ left=1in, top=1in, right=1in, bottom=1in]{geometry}
\usepackage{graphicx,bm,colonequals,amsmath,amssymb,url,xcolor,bbm}
\usepackage{array,tabularx,multirow}
\usepackage{enumitem}
\usepackage[font={footnotesize}]{caption, subcaption}
\usepackage[utf8]{inputenc}
\usepackage{enumitem}
\usepackage{soul} % for strikethrough text
\usepackage{placeins} % for FloatBarrier
\usepackage[normalem]{ulem}

\bibpunct[, ]{(}{)}{;}{a}{,}{,}
   
\usepackage{amsthm}
\newtheoremstyle{propstyle} % name
    {2mm}                    % Space above
    {1mm}                    % Space below
    {\itshape}                   % Body font
    {}                           % Indent amount
    {\scshape}                   % Theorem head font
    {.}                          % Punctuation after theorem head
    {.5em}                       % Space after theorem head
    {}  % Theorem head spec (can be left empty, meaning ‘normal’)
\theoremstyle{propstyle}
\newtheorem{proposition}{Proposition}
\theoremstyle{propstyle}
\newtheorem{definition}{Definition}
\theoremstyle{propstyle}

\theoremstyle{propstyle}

\theoremstyle{propstyle}
\newtheorem{assumption}{Assumption}

\usepackage{array,framed}
\newcounter{algorithm}
\newenvironment{algorithm}[1][]{\refstepcounter{algorithm}\par\medskip\noindent%
   \textbf{Algorithm~\thealgorithm: #1} \rmfamily}{\medskip}

\makeatletter
\renewcommand{\paragraph}{%
  \@startsection{paragraph}{4}%
  {\z@}{2ex \@plus 1ex \@minus .2ex}{-1em}%
  {\normalfont\normalsize\bfseries}%
}
\makeatother

\DeclareMathAlphabet\mathbfcal{OMS}{cmsy}{b}{n}

\newcommand{\bv}{\mathbf{v}}
\newcommand{\bc}{\mathbf{c}}

\newcommand{\bx}{\mathbf{x}}
\newcommand{\by}{\mathbf{y}}

\newcommand{\bg}{\mathbf{g}}

\newcommand{\bw}{\mathbf{w}}
\newcommand{\bS}{\mathbf{S}}
\newcommand{\bT}{\mathbf{T}}

\newcommand{\bA}{\mathbf{A}}

\newcommand{\bW}{\mathbf{W}}

\newcommand{\bG}{\mathbf{G}}
\newcommand{\bL}{\mathbf{L}}

\newcommand{\bI}{\mathbf{I}}

\newcommand{\bH}{\mathbf{H}}
\newcommand{\bU}{\mathbf{U}}
\newcommand{\bV}{\mathbf{V}}
\newcommand{\bK}{\mathbf{K}}

\newcommand{\bQ}{\mathbf{Q}}
\newcommand{\bB}{\mathbf{B}}

\newcommand{\bR}{\mathbf{R}}

\newcommand{\bfzero}{\mathbf{0}}

\newcommand{\bfmu}{\bm{\mu}}

\newcommand{\bftheta}{\bm{\theta}}
\newcommand{\bfeta}{\bm{\eta}}

\newcommand{\bfrho}{\bm{\rho}}
\newcommand{\bfSigma}{\bm{\Sigma}}

\newcommand{\bfLambda}{\bm{\Lambda}}

\newcommand{\bfOmega}{\bm{\Omega}}
\newcommand{\blt}{\tilde{\bL}}

\newcommand{\iM}{{i_1,\ldots,i_M}}
\newcommand{\jm}{{j_1,\ldots,j_m}}

\newcommand{\jk}{{j_1,\ldots,j_k}}
\newcommand{\jl}{{j_1,\ldots,j_\ell}}

\newcommand{\jM}{{j_1,\ldots,j_M}}

\newcommand{\blockdiag}{blockdiag}

\newcommand{\knots}{\mathcal{K}}

\newcommand{\order}{\mathcal{O}}
\renewcommand{\prec}{\bm{\Lambda}}
\newcommand{\pprec}{\widetilde{\bm{\Lambda}}}

\newcommand{\levol}{\mathbf{A}}
\newcommand{\grid}{\mathcal{G}}

\newcommand{\knotind}{\mathcal{K}}
\newcommand{\gridind}{\mathcal{I}}
\newcommand{\Jind}{\mathcal{J}}
\newcommand{\Sp}{\mathcal{S}}

\newcommand{\normal}{\mathcal{N}}
\newcommand{\domain}{\mathcal{D}}

\newcommand{\reals}{\mathbb{R}}

\DeclareMathOperator{\mrd}{MRD}

\title{Multi-resolution filters for massive spatio-temporal data}

\author{Marcin Jurek\thanks{Department of Statistics, Texas A\&M University} \and Matthias Katzfuss\footnotemark[1] \thanks{Corresponding author: \texttt{katzfuss@gmail.com}}}
\date{}

\begin{document}

\maketitle

\begin{abstract}

Spatio-temporal datasets are rapidly growing in size. For example, environmental variables are measured with increasing resolution by increasing numbers of automated sensors mounted on satellites and aircraft. Using such data, which are typically noisy and incomplete, the goal is to obtain complete maps of the spatio-temporal process, together with uncertainty quantification. We focus here on real-time filtering inference in linear Gaussian state-space models. At each time point, the state is a spatial field evaluated on a very large spatial grid, making exact inference using the Kalman filter computationally infeasible. Instead, we propose a multi-resolution filter (MRF), a highly scalable and fully probabilistic filtering method that resolves spatial features at all scales. We prove that the MRF matrices exhibit a particular block-sparse multi-resolution structure that is preserved under filtering operations through time. We describe connections to existing methods, including hierarchical matrices from numerical mathematics. We also discuss inference on time-varying parameters using an approximate Rao-Blackwellized particle filter, in which the integrated likelihood is computed using the MRF. Using a simulation study and a real satellite-data application, we show that the MRF strongly outperforms competing approaches.
\end{abstract}

%%%%%%%%%%%%%%%%%%%%%%%%%%%%%%%%%%%%%%%%%%%%%%%%%%%%%%%

\section{Introduction \label{sec:intro}}

% motivation
Massive spatio-temporal data have become ubiquitous in the environmental sciences, which is largely due to Earth-observing satellites providing high-resolution measurements of environmental variables on a continental or even global scale. Accounting for spatial and temporal dependence is very important for satellite data, as atmospheric variables vary over space and time, and measurements from different orbits are often complementary in their coverage. 

% SSMs and filtering perspective
When time and space are discretized, spatio-temporal data are typically modeled using a dynamical state-space model (SSM), which describes how the state (i.e., the spatial field evaluated at a spatial grid) evolves over time and how the state is related to the observations. Dynamical SSMs can include information from other sources and sophisticated temporal dynamics in terms of partial differential equations; for example, the effect of wind on atmospheric variables can be captured by an advection term. Such informative, physical evolution models are crucial for producing meaningful forecasts.
% than in the case of purely statistical models.

% filtering inference
We focus here on real-time or on-line filtering inference in linear Gaussian SSMs, which means that at each time point $t$, we are interested in the posterior distribution of the spatial field at time $t$ given all data obtained up to time $t$. The filtering distributions in this setting are Gaussian and can in principle be determined exactly by the Kalman filter \citep[][]{Kalman1960}, but this technique is not computationally feasible for large grids. Particle filter methods such as sequential importance (re-)sampling \citep[e.g.,][]{Gordon1993} are asymptotically exact as the number of particles increases, but suffer from particle collapse for finite particle size in high dimensions \citep[e.g.,][]{Snyder2008}.

% geoscience, data assimilation
In the  geosciences, filtering inference in SSMs is referred to as data assimilation \citep[see, e.g.,][for a review]{Nychka2010}, especially when the evolution is described by a complex computer model. Data assimilation is typically carried out via variational methods \citep[e.g.,][]{Talagrand1987} or the ensemble Kalman filter \citep[EnKF; e.g.,][]{Evensen1994,Evensen2007,katzfuss2016understanding,Houtekamer2016}. The EnKF represents distributions by an ensemble, which is propagated using the temporal evolution model and updated via a linear shift based on new observations. In practice, only small ensemble sizes are computationally feasible, resulting in a low-dimensional representation and substantial sampling error.

% existing approaches in statistics
In the statistics literature, computationally feasible filtering approaches for dynamical spatio-temporal SSMs often rely on low-rank assumptions \citep[e.g.,][]{Verlaan1995,Pham1998,Wikle1999,Katzfuss2010}, but such approaches cannot fully resolve fine-scale variation \citep{Stein2013a}. 
Therefore, recent methods for large spatial-only data have instead achieved fast computation through sparsity assumptions \citep[e.g.,][]{Lindgren2011a,Nychka2012,Datta2016,katzfuss2017general}, and idea that can also be used in the context of retrospective, ``off-line'' spatio-temporal analysis, in which time is essentially treated as an additional spatial dimension and the resulting spatio-temporal covariance function is modeled or approximated \citep[e.g.,][]{Zhang2014,Datta2016a}. However, most sparsity-based methods cannot be easily extended to the filtering perspective of interest here, because the sparsity structure is lost when propagating forward in time.
% Note that this filtering perspective is fundamentally different from a retrospective (``off-line'') spatio-temporal analysis, in which time is essentially treated as an additional spatial dimension. It is often straightforward to extend spatial methods to this retrospective spatio-temporal setting. 

% short overview of the MRF and its advantages
Here, we propose a novel multi-resolution filter (MRF) for big streaming spatio-temporal data based on linear Gaussian SSMs. The MRF is a highly scalable, fully probabilistic procedure that results in joint posterior predictive distributions for the spatio-temporal field of interest. 
In contrast to the EnKF, MRF computations are deterministic and do not suffer from sampling variability. 
In contrast to low-rank approaches, the MRF does not rely on dimension reduction.
% In contrast to low-rank approaches, the MRF resolves features at all spatial scales.
Similar to wavelet-based filtering methods \citep[e.g.][]{Chui1992,Cristi2000,Renaud2005,Beezley2011,Hickmann2015}, the MRF can be viewed as employing a large number of basis functions at multiple levels of spatial resolution, which can capture spatial structure from very fine to very large scales. However, as opposed to wavelets, the MRF basis functions automatically adapt to approximate the covariance structure implied by the assumed SSM.
These features allowed the MRF to strongly outperform existing approaches in our numerical comparisons.

% MRF, sparsity, and complexity
The MRF relies on a new approximate covariance-matrix decomposition, for which the resulting matrix factors exhibit a particular block-sparse multi-resolution structure. This decomposition is based on the multi-resolution approximation \citep[][]{Katzfuss2016mra,KatzfussGong} of spatial processes, which performed very well in a recent comparison of different methods for large spatial-only data \citep{Heaton2017}.
Using advanced concepts from graph theory, we prove the perhaps surprising property that the block-sparse structure of the MRF matrices can be maintained under filtering operations through time, which in turn is crucial for allowing us to show that the MRF exhibits linear computational complexity for fixed tuning parameters. Note that this is in contrast to other sparse-matrix approximations, as even matrices with simple sparsity patterns (e.g., tridiagonal matrices) do not preserve sparsity under inversion. In fact, we suspect that the multi-resolution decomposition and its special cases are unique in terms of preserving matrix sparsity.

% hierarchical matrices
We also establish a close connection between our multi-resolution decomposition and hierarchical matrices. Despite being relatively unknown in statistics, hierarchical matrices \citep[e.g.,][]{hackbusch2015hierarchical} are a highly popular and widely studied class of matrix approximations in numerical mathematics. %In contrast to the typical interpretation as a matrix approximation, we show that our statistical model directly implies a covariance matrix of hierarchical form. 
We introduce this matrix class into the statistical literature, and describe how hierarchical matrices can be applied to SSMs based on second-order partial differential equations, including those describing advection and diffusion. This marks a major step forward with respect to multiple previous hierarchical-matrix approaches for fast high-dimensional Kalman filtering \citep[e.g.][]{LiAmbikasaran2014,Saibaba2015,Ambikasaran2016}, which were only applicable in the simple case of a random walk.

% particle filter
Finally, we discuss extensions for inference on time-varying parameters that are not part of the spatial field, using a Rao-Blackwellized particle filter, in which the integrated likelihood is approximated using the MRF.

% table of contents
The remainder of this article is organized as follows. Section \ref{sec:SSMandKF} describes the linear Gaussian state-space model and reviews the Kalman filter. In Section \ref{sec:mrf}, we present the MRF. Section \ref{sec:properties} details key properties of the MRF, and Section \ref{sec:connections} discusses connections to existing approaches. Section \ref{sec:parinference} shows how the MRF can be extended when the model includes unknown parameters. In Section \ref{sec:simulations}, we present a numerical comparison of the MRF to existing methods. Section \ref{sec:data-application} demonstrates a practical application of the MRF to inferring sediment concentration in Lake Michigan based on satellite data. We conclude in Section \ref{sec:conclusion}. Proofs are given in Appendix \ref{app:proofs}.

A separate Supplementary Material document contains Sections S1-S8 with further properties, details, and proofs.  At \url{http://spatial.stat.tamu.edu}, we provide additional illustrations. All code will be provided upon publication.

%%%%%%%%%%%%%%%%%%%%%%%%%%%%%%%%%%%%%%%%%%%%%%%%%%%%%%%
\section{Spatio-temporal state-space models (SSMs) and filtering inference}\label{sec:SSMandKF}

%%%%%%%%%%
\subsection{Spatio-temporal state-space model}\label{sec:spatioTempSSM}

Let $\bx_t$ be the $n_\grid$-dimensional latent state vector of interest, representing a (mean-corrected) spatio-temporal process $x_t(\cdot)$ at time $t$ evaluated at a fine grid $\grid =\{\bg_1,\ldots,\bg_{n_\grid}\}$ on a spatial domain or region $\domain$. Further, let $\by_t$ denote the observed $n_t$-dimensional data vector at time $t$. We assume a linear Gaussian spatio-temporal state-space model given by an observation equation and an evolution equation,
\begin{eqnarray}
\by_t & = \, \bH_t \bx_t + \bv_t, \quad & \bv_t \sim \normal_{n_t}(\bfzero,\bR_t), \label{eq:obs}\\
\bx_t & = \, \levol_t\bx_{t-1} + \bw_t, & \bw_t \sim \normal_{n_\grid}(\bfzero,\bQ_t), \label{eq:evol}
\end{eqnarray}
respectively, for time $t=1,2,\ldots$. The initial state also follows a Gaussian distribution: $\bx_0 \sim \normal_{n_\grid}(\bfmu_{0|0},\bfSigma_{0|0})$. The noise covariance matrix $\bR_t$ will be assumed to be diagonal or block-diagonal here for simplicity (see Assumption \ref{ass:hr} in Section \ref{sec:sparsity}). No computationally convenient structure is assumed for the innovation covariance matrix $\bQ_t$. The observation noise $\bv_t$ and the innovation $\bw_t$ are mutually and serially independent, and independent of the state $\bx_{t-1}$. We assume that all matrices in \eqref{eq:obs}--\eqref{eq:evol} (and $\bfmu_{0|0}$ and $\bfSigma_{0|0}$) are known. The case of unknown parameters is discussed in Section \ref{sec:parinference}. 
 
The observation matrix $\bH_t$ relates the state to the observations. This enables combining observations from different instruments or modeling areal observations given by averaging over certain elements of the state vector. Here, we usually assume point-level measurements for simplicity, although a block-diagonal form for $\bH_t$ is possible (see Assumption \ref{ass:hr}). %If an observation lies exactly on one of the grid points in $\grid$, the corresponding row in $\bH_t$ consists of only zeros and a single one at the appropriate place. If an observations is not on the grid, the corresponding row of $\bH_t$ can be determined by (bi-)linear interpolation. Hence, $\bH_t$ is sparse (see Assumption \ref{ass:hr} in Section \ref{sec:sparsity}).

The evolution matrix $\levol_t$ determines how the process evolves over time. It can be specified in terms of a system of partial differential equations (PDEs), may depend on other variables, or --- in the absence of further information --- could simply be a scaled identity operator indicating a random walk over time. We assume that the evolution is local and $\levol_t$ is sparse (Assumption \ref{ass:local} in Section \ref{sec:comp-complexity}).

Note that the SSM in \eqref{eq:obs}--\eqref{eq:evol}, which is a latent Markov model of order 1, is very general and can describe a broad class of systems. Higher-order Markov models can also be written in the form \eqref{eq:obs}--\eqref{eq:evol} by expanding the state space. Non-Gaussian observations can often be transformed to be approximately Gaussian. Other extensions are also straightforward, such as letting the grid $\grid$ vary over time.

%%%%%%%%%%
\subsection{Filtering inference using the Kalman filter (KF) \label{sec:kf}}

We are interested in filtering inference on the state $\bx_t$. That is, at each time $t$, the goal is to find the conditional distribution of $\bx_t$ given all observations up to and including time $t$, denoted by $\bx_t | \by_{1:t}$, where $\by_{1:t} = (\by_1',\ldots,\by_t')'$.

For the linear Gaussian SSM in \eqref{eq:obs}--\eqref{eq:evol}, the filtering distributions are Gaussian. These filtering distributions can be obtained recursively for $t=1,2,\ldots$ using the Kalman filter \citep{Kalman1960}, which consists of a forecast step and an update step at each time point. Denote the filtering distribution at time $t-1$ by $\bx_{t-1} | \by_{1:t-1} \sim \normal_n(\bfmu_{t-1|t-1},\bfSigma_{t-1|t-1})$. The forecast step obtains the forecast or prior distribution of $\bx_t$ based on the previous filtering distribution and the evolution model \eqref{eq:evol} as
\begin{equation*}
%\label{eq:forecast}
\bx_{t} | \by_{1:t-1} \sim \normal_{n_\grid}(\bfmu_{t|t-1},\bfSigma_{t|t-1}), \qquad \bfmu_{t|t-1} \colonequals \levol_t\bfmu_{t-1|t-1}, \quad \bfSigma_{t|t-1} \colonequals \levol_t \bfSigma_{t-1|t-1} \levol_t' + \bQ_t.
\end{equation*}
Then, the update step modifies this forecast distribution based on the observation vector $\by_t$ and the observation equation \eqref{eq:obs}, in order to obtain the filtering distribution of $\bx_t$:
\begin{equation}
\label{eq:update}
\bx_{t} | \by_{1:t} \sim \normal_{n_\grid}(\bfmu_{t|t},\bfSigma_{t|t}), \qquad \bfmu_{t|t} \colonequals \bfmu_{t|t-1} + \bK_t(\by_t - \bH_t \bfmu_{t|t-1}), \quad \bfSigma_{t|t} \colonequals (\bI_{n_\grid} - \bK_t \bH_t) \bfSigma_{t|t-1},
\end{equation}
where $\bK_t \colonequals \bfSigma_{t|t-1}\bH_t'(\bH_t\bfSigma_{t|t-1} \bH_t'+\bR_t)^{-1}$ is the $n_\grid \times n_t$ Kalman gain matrix.

While the Kalman filter provides the exact solution to our filtering problem, it requires computing and propagating the $n_\grid \times n_\grid$ covariance matrix $\bfSigma_{t|t}$ and decomposing the $n_t \times n_t$ matrix $(\bH_t\bfSigma_{t|t-1} \bH_t'+\bR_t)$ in $\bK_t$, and is thus computationally infeasible for large $n_\grid$ or large $n_t$. Therefore, approximations are required for large spatio-temporal data.

%%%%%%%%%%%%%%%%%%%%%%%%%%%%%%%%%%%%%%%%%%%%%%%%%%%%%%%
\section{The multi-resolution filter (MRF) \label{sec:mrf}}

%%%%%%%%%%
\subsection{Overview}

We now propose a multi-resolution filter (MRF) for spatio-temporal SSMs of the form \eqref{eq:obs}--\eqref{eq:evol} when the grid size $n_\grid$ or the data sizes $n_t$ are large, roughly between $10^4$ and $10^9$. The MRF can be viewed as an approximation of the Kalman filter in Section \ref{sec:kf}. 

The most important ingredient of the MRF is a novel multi-resolution decomposition (MRD). Given a spatial covariance matrix $\bfSigma$, the MRD computes $\bB = \mrd(\bfSigma)$ such that $\bfSigma \approx \bB \bB'$. We will describe the MRD in detail in Section \ref{sec:mrd}. For now, we merely note that the MRD algorithm is fast, and the resulting multi-resolution factor $\bB$ is of the same dimensions as $\bfSigma$ but exhibits a particular block-sparse structure (see Figure \ref{fig:sparsity-B}).

The MRF algorithm proceeds as follows:
\begin{framed}
\begin{algorithm}
\label{alg:mrf} \textbf{Multi-resolution filter (MRF)}\\
At the initial time $t=0$, compute $\bB_{0|0} = \mrd(\bfSigma_{0|0})$. Then, for each $t=1,2,\ldots$, do:
\begin{description}[topsep=1pt,itemsep=1pt,parsep=1pt]
\item[1. Forecast Step:] Apply the evolution matrix $\levol_t$ to obtain $\bfmu_{t|t-1} = \levol_t \bfmu_{t-1|t-1}$ and $\bB_{t|t-1}^F = \levol_t\bB_{t-1|t-1}$. Carry out a multi-resolution decomposition $\bB_{t|t-1} = \mrd(\bfSigma_{t|t-1}^F)$, where $\bfSigma_{t|t-1}^F=\bB_{t|t-1}^F (\bB^F _{t|t-1})' + \bQ_t$, to obtain $\bx_t | \by_{1:t-1} \sim \normal_{n_\grid}(\bfmu_{t|t-1},\bfSigma_{t|t-1})$ with $\bfSigma_{t|t-1} = \bB_{t|t-1}\bB_{t|t-1}'$.
\item[2. Update Step:] Compute $\bB_{t|t} \colonequals \bB_{t|t-1} (\bL_t^{-1})'$, where $\bL_t$ is the lower Cholesky triangle of $\bfLambda_t \colonequals \bI_{n_\grid} + \bB_{t|t-1}' \bH_t' \bR_t^{-1} \bH_t \bB_{t|t-1}$, to obtain $\bx_t |\by_{1:t} \sim \normal_{n_\grid}(\bfmu_{t|t},\bfSigma_{t|t})$ with $\bfSigma_{t|t} = \bB_{t|t}\bB_{t|t}'$ and $\bfmu_{t|t}  = \bfmu_{t|t-1} + \bB_{t|t}\bB_{t|t}'\bH_t'\bR_t^{-1}(\by_t-\bH_t\bfmu_{t|t-1})$.
\end{description}
\end{algorithm}
\end{framed}
The key to the scalability of this algorithm is that while $\bfSigma_{t|t-1}$ and $\bfSigma_{t|t}$ are large and dense matrices, they are never explicitly calculated and instead represented by the block-sparse matrices $\bB_{t|t-1}$ and $\bB_{t|t}$, respectively. Also, as shown in Section \ref{sec:comp-cost}, $\bB_{t|t}$ has the same sparsity structure as $\bB_{t|t-1}$, which allows the cycle to start over for the next time point $t+1$.
The forecast step and update step will be discussed in more detail in Sections \ref{sec:mrfforecast} and \ref{sec:mrfupdate}, respectively.

%%%
\subsection{Details of the MRF forecast step \label{sec:mrfforecast}} 

Assume that we have obtained the filtering distribution $\bx_{t-1}|\by_{1:t-1} \sim \normal_n(\bfmu_{t-1|t-1}, \bfSigma_{t-1|t-1})$, where $\bfSigma_{t-1|t-1} =\bB_{t-1|t-1}\bB_{t-1|t-1}'$ and $\bB_{t-1|t-1}$ is a block-sparse matrix. Following the forecast step of the standard Kalman filter, we want to obtain the prior distribution at time $t$, $\mathbf{x}_{t}|\mathbf{y}_{1:t-1} \sim \normal_n(\bm{\mu}_{t|t-1}, \bm{\Sigma}_{t|t-1})$. 

Because of the sparsity of $\levol_t$ (see Assumption \ref{ass:local} in Section \ref{sec:comp-complexity}), computing the forecast mean $\bfmu_{t|t-1} = \levol_t \bfmu_{t-1|t-1}$ and the forecast basis matrix $\bB_{t|t-1}^F = \levol_t\bB_{t-1|t-1}$ is fast. Then, rather than calculating the dense $n_\grid \times n_\grid$ forecast covariance matrix $\bfSigma_{t|t-1}^F = \bB_{t|t-1}^F (\bB^F _{t|t-1})' + \bQ_t$ explicitly, we obtain its multi-resolution decomposition $\bB_{t|t-1} = \mrd(\bfSigma_{t|t-1}^F)$ as described in Section \ref{sec:mrd}. This implies an approximation to the prior covariance matrix as $\bfSigma_{t|t-1} = \bB_{t|t-1}\bB_{t|t-1}'$. Again, $\bfSigma_{t|t-1}$ does not need to be computed explicitly, because only $\bB_{t|t-1}$ is used in the update step below.

%%%
\subsection{Details of the MRF update step \label{sec:mrfupdate}}

The objective of the update step is to compute the posterior distribution
$
\bx_{t} | \by_{1:t} \sim \normal_{n_\grid}(\bfmu_{t|t},\bfSigma_{t|t})
$
given the prior quantities $\bfmu_{t|t-1}$ and $\bB_{t|t-1}$ (such that $\bfSigma_{t|t-1} = \bB_{t|t-1} \bB_{t|t-1}'$) obtained in the forecast step. 

Following the Kalman filter update in \eqref{eq:update}, we have 
\begin{align*}
\label{eq:sigtt}
\bfSigma_{t|t} 
 & = (\bI_{n_\grid} - \bK_t \bH_t) \bfSigma_{t|t-1} \\
 & = \bB_{t|t-1} \big( \bI_{n_\grid} - \bB_{t|t-1}'\bH_t'(\bH_t\bB_{t|t-1}\bB_{t|t-1}'\bH_t' + \bR_t)^{-1}\bH_t\bB_{t|t-1} \big) \bB_{t|t-1}' \\
 & = \bB_{t|t-1} \big(\bI_{n_\grid} + \bB_{t|t-1}' \bH_t' \bR_t^{-1} \bH_t \bB_{t|t-1} \big)^{-1} \bB_{t|t-1}' \\
 & = \bB_{t|t-1} \bfLambda_t^{-1} \bB_{t|t-1}' = \bB_{t|t}\bB_{t|t}',
\end{align*}
where $\bB_{t|t} \colonequals \bB_{t|t-1} (\bL_t^{-1})'$, $\bL_t$ is the lower Cholesky triangle of $\bfLambda_t \colonequals \bI_{n_\grid} + \bB_{t|t-1}' \bH_t' \bR_t^{-1} \bH_t \bB_{t|t-1}$, and we have applied the Sherman-Morrison-Woodbury formula \citep[e.g.,][]{Henderson1981} to $\bfLambda_t$. 

To obtain the filtering mean, we use the Searle set of identities \citep[][p.~151]{Searle1982}, to write the Kalman gain as 
\begin{align*}
\bK_t & = \bfSigma_{t|t-1}\bH_t'(\bH_t\bfSigma_{t|t-1} \bH_t'+\bR_t)^{-1}\\
 & = \bB_{t|t-1}\bB_{t|t-1}'\bH_t'(\bH_t\bB_{t|t-1}\bB_{t|t-1}'\bH'_t+\bR_t)^{-1}\\
 & = \bB_{t|t-1} (\bI_{n_\grid} + \bB_{t|t-1}' \bH_t' \bR_t^{-1} \bH_t \bB_{t|t-1})^{-1} \bB_{t|t-1}'\bH_t'\bR_t^{-1}\\
 & = \bB_{t|t-1} \bfLambda_t^{-1} \bB_{t|t-1}'\bH_t'\bR^{-1} = \bB_{t|t}\bB_{t|t}'\bH_t'\bR_t^{-1},
\end{align*}
and so we have 
\begin{align*}
\bfmu_{t|t} & = \bfmu_{t|t-1} + \bK_t(\by_t-\bH_t\bfmu_{t|t-1}) \\
 & = \bfmu_{t|t-1} + \bB_{t|t}\bB_{t|t}'\bH_t'\bR_t^{-1}(\by_t-\bH_t\bfmu_{t|t-1}).
\end{align*}
Thus, the MRF update step in Algorithm \ref{alg:mrf} is exact for given $\bfmu_{t|t-1}$ and $\bfSigma_{t|t-1} = \bB_{t|t-1} \bB_{t|t-1}'$.
Crucially, we will show in Proposition \ref{prop:pd} that $\bB_{t|t}$ has the same sparsity structure as $\bB_{t|t-1}$, and hence it satisfies the block-sparsity assumption at the beginning of Section \ref{sec:mrfforecast}.

%%%%%%%%%%
\subsection{The multi-resolution decomposition \label{sec:mrd}}

We now propose an approximate multi-resolution decomposition (MRD) of a generic spatial covariance matrix $\bfSigma$, which is used in the forecast step of the MRF in Algorithm \ref{alg:mrd}. Specifically, we consider a vector $\bx=\big(x(\bg_1),\ldots,x(\bg_{n_\grid})\big)' \sim \normal_{n_\grid}(\bfzero,\bfSigma)$, evaluated at a grid $\grid=\{ \bg_1, \dots, \bg_{n_\grid} \}$ over the spatial domain $\domain$. The MRD is based on a multi-resolution approximation of Gaussian processes \citep{Katzfuss2016mra} --- see Section \ref{sec:mra} for more details.

%%%
\subsubsection{Partitioning and knots \label{sec:knots}}

The MRD requires a domain partitioning and selection of knots at $M$ resolutions.
Consider a recursive partitioning of $\domain$ into $J$ regions, $\domain_1,\ldots,\domain_J$, each of which is again divided into $J$ smaller subregions (e.g., $\domain_2$ is split into subregions $\domain_{21},\ldots,\domain_{2J}$), and so forth, up to resolution $M$. We write this as
\[
\domain_{j_1,\ldots,j_{m}} = \dot{\cup}_{j_{m+1}=1}^J \, \domain_{j_1,\ldots,j_{m+1}}, \qquad (\jm) \in \{1,\ldots,J\}^m, \quad m=0,\ldots,M-1.
\]
Let $\grid_\jm = \grid \cap \domain_\jm$ 
%$\{\bg_i: \bg_i \in \domain_\jm\}$
be the grid points that lie in region $\domain_\jm$, and let $\gridind_\jm = \{i: \bg_i \in \domain_\jm \}$ be the corresponding indices, and so $\gridind=\{1, \ldots, n_\grid\}$.

\begin{figure}
	\centering
	\includegraphics[width=1.0\textwidth]{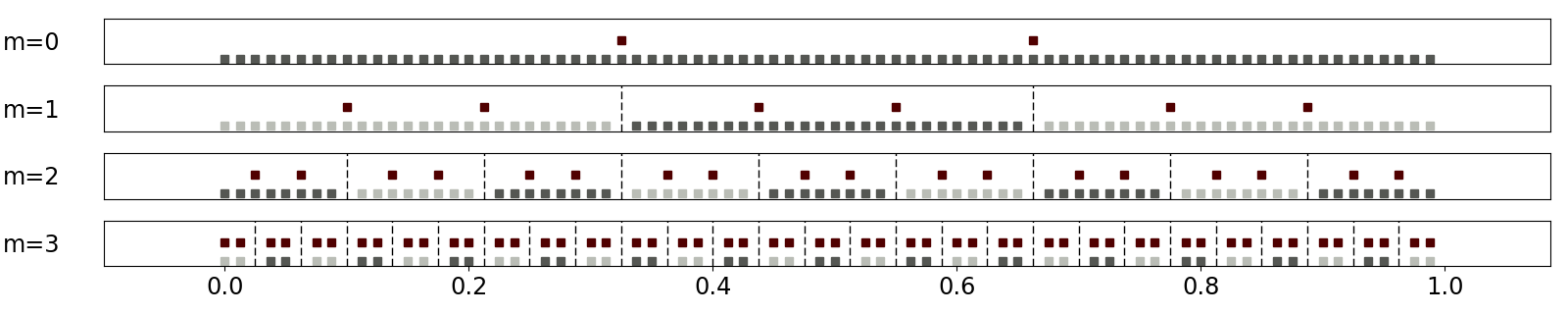}
	\caption{Illustration of knot placement for a regular grid of $n_\grid=80$ points on a one-dimensional domain $\domain$ (x-axis). Recursively for $m=0,1,\ldots,M$ (with $M=3$ here), each region is split into $J=3$ subregions (dashed lines), with $r_m=2$ knots per region (maroon dots). Grid points used as knots at resolutions $m$ are not plotted on finer ($>m$) resolutions for clarity.}
	\label{fig:knotexample}
\end{figure}

Further, we require a hierarchy of ``knot'' indices, such that $\knots_\jm$ is a small set of $r_m$ indices chosen from $\gridind_\jm$. It is assumed that for each resolution $m$, the number of knots is roughly the same in each subregion (i.e., $|\knots_\jm|=r_m$), while it may change across resolutions. We use $\knots^m=\bigcup_{\jm} \knots_\jm$ to denote the set of all knots at resolution $m$, and define $\knots^{0:m} = \bigcup_{l=0}^m \knots^l$ as the set of all knots at resolutions 0 through $m$.
To ensure that $\{\knots_\jm: (\jm) \in \{1,\ldots,J\}^m; m=0,1,\ldots,M \}$ is a partition of $\{1,\ldots,n_\grid\}$, we sequentially choose $\knots_\jm \subset (\gridind_\jm \setminus \knots^{0:m-1})$ for $m=0,1,\ldots,M$. 

In practice, we often choose $J=2$ or $J=4$. Each $r_{m-1}$ should be sufficiently large to capture the dependence between the $\domain_\jm$ that is not already captured at lower resolutions, which often means that $r_m$ can decrease as a function of $m$. Each set of knots $\knots_\jm$ could be chosen as a roughly uniform grid over the subregion $\domain_\jm$. The partitioning and knot selection is illustrated in a toy example in Figure \ref{fig:knotexample}.

Note that because $\grid$ is assumed constant over time here, we only need to do this partitioning and selection of knots once (not at each time point). We also assume that the elements in $\bx_t$ are ordered such that if $(\jM) \succ_L (i_1, \dots i_M)$, where $\succ_L$ stands for lexicographic ordering, then $\min\left(\gridind_\jM\right) > \max\left(\gridind_{i_1, \dots, i_M}\right)$.

%%%
\subsubsection{The MRD algorithm}

For index sets $\Jind_1$ and $\Jind_2$, denote by $\bfSigma[\Jind_1,\Jind_2]$ the submatrix of $\bfSigma$ obtained by selecting the $\Jind_1$ rows and $\Jind_2$ columns, and $\bfSigma[\Jind_1,:\,]$ is the submatrix of the $\Jind_1$ rows and all columns. Based on grid indices $\{\gridind_\jm\}$ and knot indices $\{\knotind_\jm\}$ selected as described in Section \ref{sec:knots}, the multi-resolution decomposition of a spatial covariance matrix $\bfSigma$ proceeds as follows:
\begin{framed}
\begin{algorithm}
\label{alg:mrd} \textbf{Multi-resolution decomposition of $\bfSigma$}\\
For $m=0,1,\ldots,M$ and $(\jm) \in \{1,\ldots,J\}^m$:
\begin{itemize}[topsep=1pt,itemsep=1pt,parsep=1pt]
\item For $\ell=0,\ldots,m$, compute
\begin{align}
\bW^\ell_\jm & = \textstyle\bfSigma[\gridind_\jm,\knotind_\jl] - \sum_{k=0}^{\ell-1} \bW_\jm^k(\bV_\jk^k)^{-1} (\bV_\jl^k)' \label{eq:W}\\ 
\bV^\ell_\jm & = \textstyle\bfSigma[\knotind_\jm,\knotind_\jl] - \sum_{k=0}^{\ell-1} \bV_\jm^k(\bV_\jk^k)^{-1} (\bV_\jl^k)'. \label{eq:V}
\end{align}
\item Set $\bB_\jm = \bW^m_\jm (\bV_\jm^m)^{-1/2}$.
\end{itemize}
Return $\bB = \mrd(\bfSigma)$, where $\bB = \left( \bB^M,\bB^{M-1},\ldots,\bB^0 \right)$ with $\bB^m = \blockdiag(\{\bB_\jm: (\jm) \in \{1,\ldots,J\}^m \})$.
\end{algorithm}
\end{framed}
The resulting matrix $\bB$ is of the same size as $\bfSigma$ but has a block-sparse structure, which is illustrated in Figure \ref{fig:sparsity-B}.

\begin{figure}
	\begin{subfigure}{.32\textwidth}
	\centering
	\includegraphics[width = 0.8\linewidth]{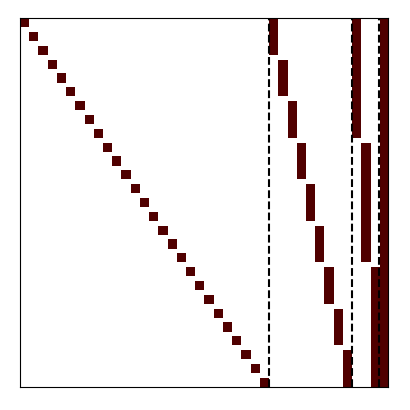}
	\caption{$\bB = \mrd(\bfSigma)$}
	\label{fig:sparsity-B}
	\end{subfigure}
	\begin{subfigure}{.32\textwidth}
	\centering
	\includegraphics[width=0.8\linewidth]{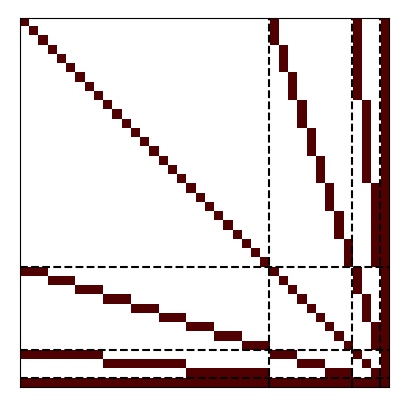}
	\caption{$\bB'\bB$}
	\label{fig:sparsity-BB}
	\end{subfigure}
	\begin{subfigure}{.32\textwidth}
	\centering
	\includegraphics[width = 0.8\linewidth]{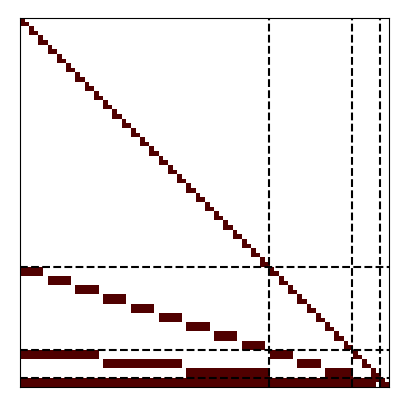}
	\caption{$\bL = \text{chol}(\bB'\bB)$ and $\bL^{-1}$}
	\label{fig:sparsity-L}
	\end{subfigure}
	\caption{Sparsity patterns for $n_\grid=80$, $M=3, \, J=3, \,$ and $r_m=2$ for $m=0, \dots, 3$. Rows and columns correspond to hierarchically arranged elements of the grid $\mathcal{G}$ in Figure \ref{fig:knotexample} from resolution $m=3$ down to $m=0$. Blocks corresponding to different resolutions are separated by dashed lines. Relating Figure \ref{fig:sparsity-B} to Figure \ref{fig:priorbf}, we see the transition from many sparse columns corresponding to many compactly supported basis functions ($m=3$), to only two dense columns corresponding to two basis functions supported over the entire domain ($m=0$).}
	\label{fig:sparsity}
\end{figure}

%%%%%%%%%%%%%%%%%%%%%%%%%%
\section{Properties of the multi-resolution filter} \label{sec:properties}

%%%%%%
\subsection{Approximation accuracy\label{sec:accuracy}}

The only difference between the MRF (Algorithm \ref{alg:mrf} and the exact Kalman filter (Section \ref{sec:kf}) is the MRD approximation of the forecast covariance matrix at each time point; that is, instead of taking $\bfSigma_{t|t-1} = \bfSigma_{t|t-1}^F$, the MRF assumes $\bfSigma_{t|t-1} = \bB_{t|t-1}\bB_{t|t-1}'$ with $\bB_{t|t-1} = \mrd(\bfSigma_{t|t-1}^F)$.
Hence, the MRF is exact when the MRD at each time point is exact. 

However, the MRD is only exact in certain special cases. One trivial example is given by $M=0$ and $r_0=n_\grid$ (see Section S1). Thus, the MRF converges to the exact Kalman filter as $r_0 \rightarrow n_\grid$, but computational feasibility for large $n_\grid$ relies on $r_0 \ll n_\grid$. Another instance of exactness is when $\bfSigma_{t|t-1}^F$ is based on an exponential covariance function on a one-dimensional domain, $\domain \subset \mathbb{R}$, and we place a total of $r_m=J-1$ knots, one at each subregion boundary \citep[][Prop.~6]{KatzfussGong}. 
Figure \ref{fig:knotexample} provides an example of such knot placement.
Finally, approximation error can also be avoided when $\levol_t = c_t \bI_{n_\grid}$ with $c_t \in \reals^+$ and $\bQ_t=\bfzero$. In this case we can set $\bB_{t|t-1} \colonequals \sqrt{c_t} \,\bB_{t-1|t-1}$, rather than employing the MRD in the forecast step. In data assimilation, the assumption $\bQ_t=\bfzero$ is quite common, when model error is incorporated through multiplicative inflation of the forecast covariance matrix.

Aside from these special cases, the MRD and hence the MRF are approximate. However, the MRA, which is the technique underlying the MRD (see Section \ref{sec:mra}), can provide excellent accuracy, as has been shown, for example, by \citet{Katzfuss2016mra}, \citet{KatzfussGong} and in a recent comparison of different methods for large spatial data \citep{Heaton2017}. In applications where accuracy is crucial, one could successively increase the number of knots $r_m$ used at low resolutions until the inference ``converges.'' We demonstrate the MRF's accuracy numerically in Sections \ref{sec:simulations} and \ref{sec:data-application}. In practice, the SSM in \eqref{eq:obs}--\eqref{eq:evol} will usually be an approximation to the true system, and we expect the MRD approximation error to often be negligible relative to the error due to model misspecification.

%%%%%%%%%%%%%
\subsection{Computational complexity \label{sec:comp-cost}}

We now determine the memory and time complexity of the MRF algorithm under the assumption that $n=\order(n_\grid) = \order(n_t)$ for all $t=1,2,\ldots$. We  also define $N\colonequals \sum_m r_m$.

%%%%
\subsubsection{Sparsity and memory requirements \label{sec:sparsity}}

As can be seen in Algorithm $\ref{alg:mrd}$, a multi-resolution factor is composed of block-diagonal submatrices by construction. The following proposition quantifies the number of its nonzero elements. 
\begin{proposition}
\label{prop:sparsity-B}
For a covariance matrix $\bfSigma$, each row of $\bB = \mrd(\bfSigma)$ has $N$ nonzero elements.
\end{proposition}
Thus, if $r_m=r$ for $m=0, \dots, M$, then each row of $\bB$ has exactly $(M+1)r$ nonzero elements. Figure \ref{fig:sparsity-B} illustrates this case for $M=3, J=3$ and $r=2$. 
The MRD results in a convenient structure of the inner product of the multi-resolution factor.  The following statement describes the sparsity pattern of this inner product (see Figure \ref{fig:sparsity-BB}), while Proposition $\ref{prop:pd}$ shows its usefulness in applications to filtering problems. 
\begin{proposition}
\label{prop:blockBB}
Let $\bB=\mrd(\bfSigma)$ for some covariance matrix $\bfSigma$. Then $\bB'\bB$ is a block matrix with $M+1$ row blocks and $M+1$ column blocks. For $k,l = 0,\ldots,M$ with $k\geq l$, the $(M+1-k,M+1-l)$-th block is of dimension $|\knots^k| \times |\knots^l|$ and is itself block-diagonal with blocks that are $r_l$ columns wide.
\end{proposition}
The following technical assumption ensures that both $\bH_t$ and $\bR_t$ are block-diagonal with blocks corresponding to indices $\gridind_{\jM}$ within each of the finest subregions:
\begin{assumption}
\label{ass:hr}
Let $i \in \gridind_\iM$ and $j \in \gridind_\jM$. Assume $\bR_t[i,j]=0$, unless $(\iM) = (\jM)$. Further, if $\bH_t[i,j] \neq 0$, then $\bH_t[i,k] = 0 $ for all $k \notin \gridind_\jM$. Finally, if $i_1,i_2 \in \gridind_\jM$ and $i_1<i_2$, then for all $i_3$ with $i_1 < i_3 < i_2$, we have $i_3 \in \gridind_\jM$.
\end{assumption}
 This assumption guarantees the key property of the MRF: The sparsity pattern of the multi-resolution factor is preserved in the update step; that is, $\bB_{t|t} \in \Sp(\bB_{t-1|t-1})$, where $\Sp(\bG)$ denotes the set of matrices whose set of structural zeros is the same or a superset of the structural zeros in some matrix $\bG$. We also use $\bG^L$ to denote the lower triangle of $\bG$, meaning that
 $
 \textstyle{\bG^L[i,j] = \bG[i,j] \text{ if } i\geq j \text{, and } \bG^L[i,j] = 0 \text{ otherwise}.}
 $
\begin{proposition} 
\label{prop:pd}
Let $\bB_{t|t-1}, \bB_{t|t}, \prec_t, \bL_t$ be defined as in Algorithm \ref{alg:mrf}. Under Assumption \ref{ass:hr}, we have:
\begin{enumerate}
	\item $\prec_t \in \Sp(\bB_{t|t-1}'\bB_{t|t-1})$;
	\item $\bL_t \in \Sp(\prec_t^L)$ and $\bL_t^{-1} \in \Sp(\prec_t^L)$;
	\item $\bB_{t|t} \in \Sp(\bB_{t|t-1})$.
\end{enumerate}
\end{proposition}

We state one more proposition that proves useful in determining the computational complexity:
\begin{proposition}
\label{prop:sparsity}
If $\bL_t$ is the lower Cholesky factor of $\prec_t$, then each column of $\bL_t$ has at most $\order(N)$ nonzero elements.
\end{proposition}
\noindent
Figure \ref{fig:sparsity-L} illustrates the structure of $\bL$. 

The results above show that all matrices computed in the MRF Algorithm \ref{alg:mrf} are very sparse, with only $\order(n N)$ nonzero entries. The update step preserves the sparsity, so that $\bB_{t|t} \in \Sp(\bB_{t|t-1})$. Due to the Markov structure of order 1 implied by our state-space model, there is no need to store matrices from previous time points, and so the memory complexity of the entire MRF algorithm is $\order(n N)$.

%%%%%%%%%%%%%
\subsubsection{Computation time \label{sec:comp-complexity}}

For determining the time complexity of the MRF, we assume that the number of knots within each subregion is constant across resolutions (i.e. $r_m=r$ for $m=0, \dots, M$) and so $N=(M+1)r$. While the efficacy of our method does not depend on this assumption, it greatly simplifies the complexity calculations and helps to develop an intuition regarding its computational benefits. 

\begin{proposition}\label{prop:mrd-comp}
Given a covariance matrix $\bfSigma$, $\bB = \mrd(\bfSigma)$ can be computed in $\mathcal{O}(nN^2)$ time using Algorithm \ref{alg:mrd}.
\end{proposition}

We further assume that the evolution is local, in the sense that the nonzero elements in any given row of $\levol_t$ only correspond to grid points in a small number of regions at the finest resolution of the domain partitioning:
\begin{assumption}
\label{ass:local}
Assume that the evolution matrix $\levol_t$ is sparse with at most $\mathcal{O}(r)$ nonzero elements per row, which must only correspond to a small, constant number of subregions,
$$
|\{ \gridind_\jM : \exists j \in \gridind_\jM \text{ such that } \levol_t[i,j] \neq 0 \}|  \leq c, \quad i=1,\ldots,n.
$$
\end{assumption}
For example, for local evolution in two-dimensional space, we have $c \leq 4$.

\begin{proposition}\label{prop:mrf-complex}
Under Assumptions \ref{ass:hr} and \ref{ass:local}, the MRF in Algorithm \ref{alg:mrf} requires $\mathcal{O}(nN^2)$ operations at each time step $t$.
\end{proposition}

In practice, $N = (M+1)r$ is chosen by the user depending on the required approximation accuracy and the available computational resources. For fixed $N$, the time and memory complexity of Algorithm \ref{alg:mrf} are linear in $n$. If $M$ increases as $M=\mathcal{O}(\log n)$ for increasing $n$ \citep[e.g.,][]{Katzfuss2016mra} and $r$ is held constant, the resulting complexity is quasilinear.

%%%%%%%%%%%%%
\subsection{Distributed computation \label{sec:distr}}

For truly massive dimensions (i.e., $n_\grid = \order(10^7)$ or more), memory limitations will typically require distributing the analysis across several computational nodes. The MRF is well suited for such a distributed environment, as information pertaining to different subregions of the domain can be stored and processed in separate nodes, with limited communication overhead required between nodes.
We plan to leverage these properties of the MRF by designing an implementation of Algorithm \ref{alg:mrf} that can be deployed in a high-performance-computation environment. We include further details in Section S2.

%%%%%%%%%%%%%
\subsection{Forecasting and smoothing \label{sec:extensions}}

Forecasting is straightforward using the MRF. Given the filtering distribution $\bx_T|\by_{1:T}$ as obtained by Algorithm \ref{alg:mrf}, we can compute the $k$-step-ahead forecast  $\bx_{T+k}|\by_{1:T}$ by simply carrying out the forecast step in Algorithm \ref{alg:mrf} $k$ times, while skipping the update step. More precisely, we carry out Algorithm \ref{alg:mrf} for $t=T+1,\ldots,T+k$, but at each time point $t$, we replace Step 2 by simply setting $\bfmu_{t|t}  = \bfmu_{t|t-1}$ and $\bB_{t|t} = \bB_{t|t-1}$. The accuracy of such forecasts will depend heavily on the quality of the evolution matrices $\bA_t$, and so a physics-informed evolution can result in much better forecasts than simple models such as random walks.

In some applications, one might also be interested in obtaining retrospective smoothing distributions $\bx_t|\by_{1:T}$ for $t<T$. These can be computed exactly by carrying out the Kalman filter up to time $T$, and then carrying out recursive backward smoothing \citep[e.g.,][]{rauch1965maximum}, but this is not feasible for large grids. It is challenging to extend the MRF by deterministic backward-smoothing operations that preserve sparsity, but it may be possible to devise a scalable MRF-based forward-filter-backward-sampler algorithm. We intend to investigate this modification in future work.

%%%%%%%%%%%%%%%%%%%%%%%%%%
\section{Connections to existing methods\label{sec:connections}}

In this section, we discuss in some detail the connections between our MRF and hierarchical matrix decompositions and basis-function approximations. Further, in Section S3, we discuss connections to multi-resolution autoregressive models, which demonstrate that the MRF can also be interpreted as a nested Kalman filter that proceeds over resolutions within each outer filtering step over time.

%%%%%%
\subsection{MRD as hierarchical low-rank decomposition \label{sec:HODLR}}

Hierarchical off-diagonal low-rank (HODLR) matrices are a popular tool in numerical analysis, and they have recently also been applied to Gaussian processes \citep[e.g.][]{Ambikasaran2013, Ambikasaran2016}. In HODLR matrices, the off-diagonal blocks are recursively specified or approximated as low-rank matrices. In this section, we show the connection between the HODLR format and the MRD when $J=2$. 

\begin{definition} \emph{\citep{Ambikasaran2016}}
A matrix $\bK \in \mathbb{R}^{N\times N}$ is termed a 1--level hierarchical off-diagonal low-rank (HODLR) matrix of rank $p$, if it can be written as
\[
\bK = \left[ \begin{array}{cc} 
	\bK_1 ^{(1)} & \bU_1^{(1)} (\bV_{1} ^{(1)})'  \\
	\bU_2 ^{(1)} (\bV_{2}^{(1)})' & \bK_2 ^{(1)} 
	\end{array} \right],
\]
where $\bK_i^{(1)} \in \reals^{N/2 \times N/2}$, and $\bU_i^{(1)}, \bV_i^{(1)}\in \reals^{N/2\times p}$. We call $\bK$ an $m$--level HODLR matrix of rank $p$ if both diagonal blocks are $(m-1)$--level HODLR matrices of rank $p$.
\end{definition}

If we use $H_m^p$ to denote the set of all $m$-level HODLR matrices of rank $p$, then it follows that $H_m^p \subset H_{m-1}^p$.
 The optimal low-rank representation is obtained by specifying the matrices $\bU_i^{(j)}$ and $\bV_i^{(j)}$ as the first $p$ singular vectors of the corresponding off-diagonal submatrix \citep[Item 5.6.13]{hogben2006handbook}, but this is prohibitively expensive. \cite{Ambikasaran2016} discuss multiple ways of approximating this low-rank representation.

We now show that the outer product of an MRD factor is a HODLR matrix, specifically one in which the low-rank approximations are obtained as skeleton factorizations.
\begin{proposition} \label{prop:hodlr}
Let $\bB=\mrd(\bfSigma)$, where the decomposition is based on a partitioning scheme with $J=2$ and $r_m=r$ for $m=0, \dots, M$. Then, $\bB\bB' \in H_M^r$.
\end{proposition}
The proof is given in Appendix \ref{app:proofs}. It can easily be extended to $r_m$ varying by resolution. Thus, the MRF approximation of the prior covariance matrix, $\bfSigma_{t|t-1} = \bB_{t|t-1}\bB_{t|t-1}'$, is a HODLR matrix \citep{Ambikasaran2016}. In contrast to previous approaches using HODLR matrices for spatio-temporal models \citep[e.g.][]{LiAmbikasaran2014,Saibaba2015}, the block-sparse MRD matrices allow the MRF to handle non-diagonal evolution matrices $\levol_t$ and full-rank model-error matrices $\bQ_t$.

%%%%%%
\subsection{MRD as basis-function approximation} \label{sec:mra}

The MRD is related to the multi-resolution approximation \citep[MRA;][]{Katzfuss2016mra} of a Gaussian process as a weighted sum of increasingly compactly supported basis functions at $M$ resolutions. While the MRD adapts the MRA to an approximate decomposition of a covariance matrix evaluated at a spatial grid, $\bfSigma=\bB\bB'$, we can similarly interpret each column of $\bB$ as a basis vector over the grid. In other words, the spatial field $\bx \sim \normal(\bfzero,\bfSigma)$ is approximated as $\bx \approx \bB \bfeta$, where $\bfeta \sim \normal(\bfzero,\bI)$ is the vector of independent standard normal weights. By interpolating the values of the basis vectors between grid points, we can visualize the columns of $\bB$ as basis functions, which is illustrated in Figure \ref{fig:priorbf}.

The basis functions obtained in this way exhibit interesting properties. Their support is increasingly compact as the resolution increases, and basis functions at low resolution capture the large-scale structure. There are strong connections between the MRD and stochastic wavelets, with the major difference that the shape of the basis functions in the MRD adapts to the covariance structure in $\bfSigma$. This adaptation is especially useful in the spatio-temporal context here, which requires approximation of the forecast covariance matrix $\bfSigma_{t|t-1}^F$ that depends on the data at previous time points and is hence highly nonstationary. The compact support stems from the assumption of a block-sparse structure at each resolution in the MRD, which is equivalent to assuming that the remainder at each resolution is conditionally independent between subregions at that resolution, given the terms at lower resolutions. In general, this assumption is not satisfied and thus produces an approximation error, although the MRD is exact in some special cases (see Section \ref{sec:accuracy}).

\begin{figure}
	\centering
	\includegraphics[width=1.0\linewidth]{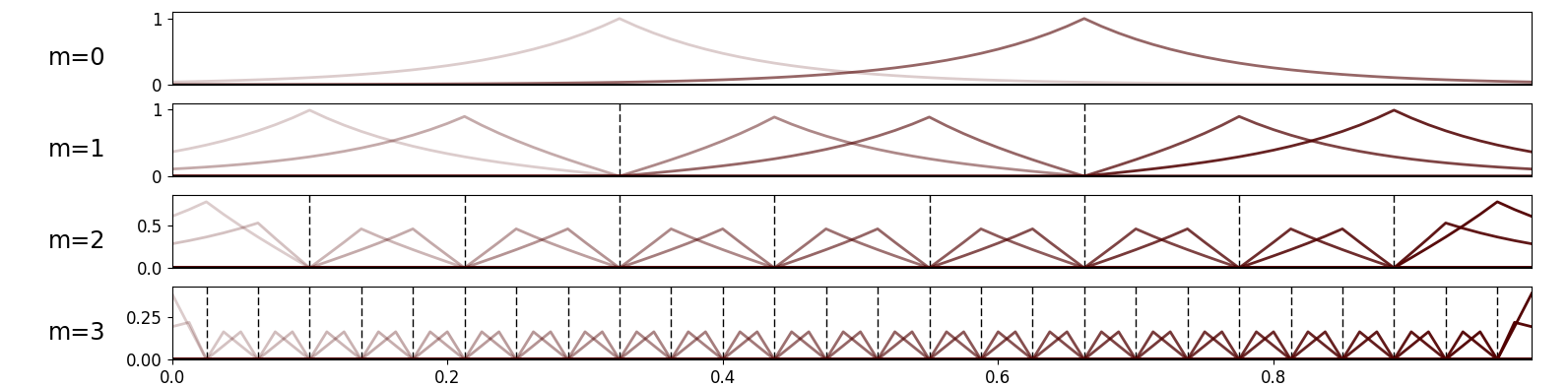}
	\caption{Basis functions obtained by interpolating the entries in each column of $\bB = \mrd(\bfSigma)$ in Figure \ref{fig:sparsity-B} using the grid from Figure \ref{fig:knotexample}, with $\bfSigma$ based on an exponential covariance with range $0.3$. Each basis function's support is restricted to one of the subregions (dashed lines) at each resolution.}
	\label{fig:priorbf}
\end{figure}

%%%%%%%%%%%%%%%%%%%%%%%%%%
\section{Parameter inference \label{sec:parinference}}

If there are random, time-varying parameters $\bftheta_t$ in any of the matrices in \eqref{eq:obs}--\eqref{eq:evol}, that are distinct from the Gaussian state $\bx_t$, we can make inference on these parameters using an approximate Rao-Blackwellized particle filter \citep{Doucet2000}, in which we use the MRF algorithm to approximately integrate out the high-dimensional state $\bx_t$ at each time point. An alternative approach, based on including the unknown parameters in the state vector, otherwise known as data augmentation, tends to work poorly for certain parameters and thus is less general \citep[e.g.,][]{DelSole2010,Katzfuss2017c}.

To derive our filter, note that we have
\[
p(\by_{1:T} | \bftheta_{1:T}) = \prod_{t=1}^T p(\by_t|\by_{1:t-1},\bftheta_{1:t}) \equalscolon \prod_{t=1}^T \mathcal{L}_t(\bftheta_{1:t}),
\] 
where, after integrating out $\bx_t$, we have $\by_t|\by_{1:t-1},\bftheta_{1:t} \sim \normal_{n_t}(\bH_t\bfmu_{t|t-1},\bH_t\bfSigma_{t|t-1}\bH_t'+\bR_t)$ with $\bfSigma_{t|t-1} = \bB_{t|t-1}\bB_{t|t-1}'$. By applying a matrix determinant lemma \citep[e.g.,][Thm.~18.1.1]{Harville1997} to the determinant and the Sherman-Morrison-Woodbury formula to the quadratic form in the multivariate normal density, it is straightforward to show that the integrated filtering likelihood at time $t$ can be written, up to a constant, as
\begin{equation}
\label{eq:lik}
 - 2 \log \mathcal{L}_t(\bftheta_{1:t}) = 2 \log |\bL_t| + \log|\bR_t| + (\by_t-\bH_t\bfmu_{t|t-1})'\bR_t^{-1}(\by_t-\bH_t\bfmu_{t|t-1}) - \tilde\by_t'\tilde\by_t,
\end{equation}
where $\tilde\by_t \colonequals \bB_{t|t}'\bH_t'\bR_t^{-1}(\by_t-\bH_t\bfmu_{t|t-1})$, and we have omitted dependence on the parameters $\bftheta_{1:t}$ for the terms on the right-hand side.

Assuming that the priors for the $\bftheta_t$ are given by $p_0(\bftheta_0)$ for $t=0$, and recursively by $p_t(\bftheta_{t} |\bftheta_{t-1})$ for $t=1,2,\ldots$, the particle MRF proceeds as follows:
\begin{framed}
\begin{algorithm}
\label{alg:mrfpf} \textbf{Particle MRF}\\
At time $t=0$, for $i=1,\ldots,N_p$, draw $\bftheta_0^{(i)} \sim p_0(\bftheta_0)$ with equal weights $w_0^{(i)}=1/N_p$, and compute $\bfmu_{0|0}^{(i)} = \bfmu_{0|0}(\bftheta^{(i)})$ and $\bB_{0|0}^{(i)} = \mrd(\bfSigma_{0|0}(\bftheta_0^{(i)}))$. Then, for each $t=1,2,\ldots$, do:
\begin{itemize}
	\item For $i=1,\ldots,N_p$:
		\begin{itemize}
			\item Sample $\bftheta_{t}^{(i)}$ from a proposal distribution $q_t(\bftheta_t|\bftheta_{t-1}^{(i)})$. 
			\item Forecast step: Compute $\bfmu_{t|t-1}^{(i)} = \levol_t(\bftheta_t^{(i)}) \bfmu_{t-1|t-1}^{(i)}$, $\bB_{t|t-1}^F{}^{(i)} = \levol_t(\bftheta_t^{(i)}) \bB_{t-1|t-1}^{(i)}$, and $\bB_{t|t-1}^{(i)} = \mrd(\bfSigma_{t|t-1}^{(i)})$, where $\bB_{t|t-1}^F{}^{(i)} (\bB^F _{t|t-1}{}^{(i)})' + \bQ_t(\bftheta_t^{(i)}))$.
			\item Update step: Compute $\bfLambda_t^{(i)} = \bI_{n_\grid} + \bB_{t|t-1}^{(i)}{}' \bH_t(\bftheta_t^{(i)})' \bR_t(\bftheta_t^{(i)})^{-1} \bH_t(\bftheta_t^{(i)}) \bB_{t|t-1}^{(i)}$, $\bL_t^{(i)}$ as the lower Cholesky triangle of $\bfLambda_t^{(i)}$, $\bB_{t|t}^{(i)} = \bB_{t|t-1}^{(i)} ((\bL_t^{(i)})^{-1})'$, and $\bfmu_{t|t}^{(i)}  = \bfmu_{t|t-1}^{(i)} + \bB_{t|t}^{(i)}\bB_{t|t}^{(i)}{}'\bH_t(\bftheta_t^{(i)})'\bR_t(\bftheta_t^{(i)})^{-1}(\by_t-\bH_t(\bftheta_t^{(i)})\bfmu_{t|t-1}^{(i)})$.
			\item Using the quantities just computed for $\bftheta_t^{(i)}$, calculate $\mathcal{L}_t(\bftheta_{1:t}^{(i)})$ as in \eqref{eq:lik}, and update the particle weight $w^{(i)}_t \propto w^{(i)}_{t-1} \mathcal{L}_t(\bftheta_{1:t}^{(i)}) p_t(\bftheta_{t}^{(i)} |\bftheta_{t-1}^{(i)})/ q_t(\bftheta_{t}^{(i)} |\bftheta_{t-1}^{(i)})$.
		 \end{itemize}
	 \item The filtering distribution is $p(\bftheta_t,\bx_t| \by_{1:t}) = \sum_{i=1}^{N_p} w_t^{(i)} \delta_{\bftheta_t^{(i)}}(\bftheta_t) \, \normal_{n_\grid}(\bx_t|\bfmu_{t|t}^{(i)},\bB_{t|t}^{(i)}\bB_{t|t}^{(i)}{}')$.
 \item If desired, resample the triplets $\{(\bftheta_{t}^{(i)},\bfmu_{t|t}^{(i)},\bB_{t|t}^{(i)}): i=1,\ldots,N_p\}$ with weights $w_t^{(1)},\ldots,w_t^{(N_p)}$, respectively, to obtain an equally weighted sample \citep[see, e.g.,][for a comparison of resampling schemes]{Douc2005}.  
\end{itemize}
\end{algorithm}
\end{framed}

Section S5 presents numerical experiments demonstrating the accuracy of Algorithm \ref{alg:mrfpf} and its advantage over a low-rank particle filter.

%%%%%%%%%%%%%%%%%%%%%%%%%%%%%%%%%%%%%%%%%%%%%%%%%%%%%%%
\section{Simulation study \label{sec:simulations}}

We used simulated data to compare the performance of the MRF with several filtering methods:
\begin{description}[itemsep=0mm,topsep=1mm]
\item[KF:] The Kalman filter (see Section \ref{sec:kf}) provides the exact filtering distributions, but has $\order(n^3)$ time complexity at each time point.
\item[MRF:] The multi-resolution filter proposed here in Section \ref{sec:mrf}, with $\mathcal{O}(nN^2)$ time complexity, where $N=\sum_{m=0}^M r_m$.
\item[EnKF:] An ensemble Kalman filter with stochastic updates \citep[e.g.,][Sect.~3.1]{katzfuss2016understanding}. We set the ensemble size to $N$ and use Kanter's function \citep{Kanter} for tapering such that the resulting matrix has $N$ nonzero entries per row. This results roughly in $\mathcal{O}(nN^2)$ time complexity \citep[e.g.,][]{Tipp:Ande:Bish:03}.
\item[LRF:] A low-rank-plus-diagonal filter that can be viewed as a spatio-temporal extension of the modified predictive process \citep{Finley2009} and as a special case of a fixed-rank filter \citep{Cressie2010a}. Moreover, it can be viewed as a special case of the MRF (hence allowing for ease of comparison) with $M=1$ resolutions and $N$ knots at resolution 0, where each grid point is in its own partition at resolution 1, resulting in a time complexity of $\order(n N^2)$.
\item[MRA:] The MRA \citep{Katzfuss2016mra} is a spatial-only method. It can essentially be viewed as a special case of the MRF, for which the filtering distribution at each time $t$ is obtained by assuming that only $\by_t$ and no data at previous time points are available. It has the same $\order(n N^2)$ complexity as the MRF.
\end{description}
While the KF provides the exact filtering distributions, it is only computationally feasible due to the deliberately small problem size chosen here. All other methods attempt to approximate the exact KF solution, but have the advantage of being scalable to much larger grid sizes. For a fair comparison, all approximate methods used the same $N$, which trades off approximation accuracy and computational complexity. Further, we acknowledge that the EnKF was designed for nonlinear evolution in operational data assimilation, and it is thus more widely applicable than the other methods.

We used two criteria to compare the performance of the approximate filters: the Kullback-Leibler (KL) divergence between the true and approximated filtering distribution of the state vector (i.e., the joint distribution for the entire spatial field), and the ratio of the root mean squared prediction error (RMSPE) achieved by each approximate method relative to the RMSPE of the KF. Detailed definitions of the criteria can be found in Section S4.1. Lower is better for both criteria, with optimal values of 0 for the KL divergence and 1 for the RMSPE ratio. In addition, Section S4.4 examines the performance of all methods in terms of the confidence-interval coverage. All quantities were averaged over 50 simulated datasets.

\subsection{One-dimensional circular domain}\label{sec:simul-1d}

In our first simulation scenario, we considered a diffusion-advection model on a one-dimensional domain consisting of a circle with a unit circumference. After discretizing both the spatial and the temporal dimensions using $n_\grid=80$ and $T=20$ regularly spaced points, respectively, we obtained a linear model as in \eqref{eq:obs}--\eqref{eq:evol}, where $\levol_t$ was a tri-diagonal matrix and $\bQ_t = \sigma_w^2 \left[ \mathcal{M}_{\nu, \lambda}(s_i, s_j) \right]_{i,j=1, \dots n_\grid}$ was based on a Mat\'ern correlation function $\mathcal{M}_{\nu,\lambda}(\cdot,\cdot)$ with smoothness $\nu$ and range $\lambda$.
At each time point, we randomly selected $n_t$ observed locations, so that $\bH_t$ is a subset of the identity, and we set $\bR_t = \sigma_v^2\bI_{n_t}$. A detailed description of the simulation, including examples of process realizations, is given in Section S4.2.

\begin{table}
	\centering
	\small
	\begin{tabular}{c|ccccc} 
	& $n_t/n_\grid$ & $\nu$ & $\lambda$ & $\sigma^2_w$ & $\sigma^2_v$ \\
	\hline \hline
	baseline & 0.3 & 0.5 & 0.1 & 0.5 & 0.05 \\
	\hline
	smooth & 0.3 & \textbf{1.5} & 0.1 & 0.5 & 0.05 \\
	dense obs.\ & \textbf{0.8} & 0.5 & 0.1 & 0.5 & 0.05 \\
	low noise & 0.3 & 0.5 & 0.1 & 0.5 & \textbf{0.01}
	\end{tabular}
	 \caption{Settings used in the \textbf{1D} simulation. Bold values indicate changes with respect to the baseline.}
	 \label{tab:baseline}
\end{table}

Because of the many possible choices of parameters, we first established baseline settings that we considered relevant for practical applications, and then examined the effects of changing them one by one. The resulting simulation scenarios are detailed in Table \ref{tab:baseline}. 
For the MRD, we set $M=3$, $J=3$, and $r_m=2$ for all $m$, and so we used $N=(3+1)2=8$ for EnKF, LRF, and MRA.
 
As shown in Figure \ref{fig:scores1d}, the MRF performed best in all four scenarios, both in terms of the KL divergence and the RMSPE ratio.

\begin{figure}
	\begin{subfigure}{1.0\textwidth}
		\includegraphics[trim={0, 3mm, 0, 0}, clip, width=1.0\textwidth]{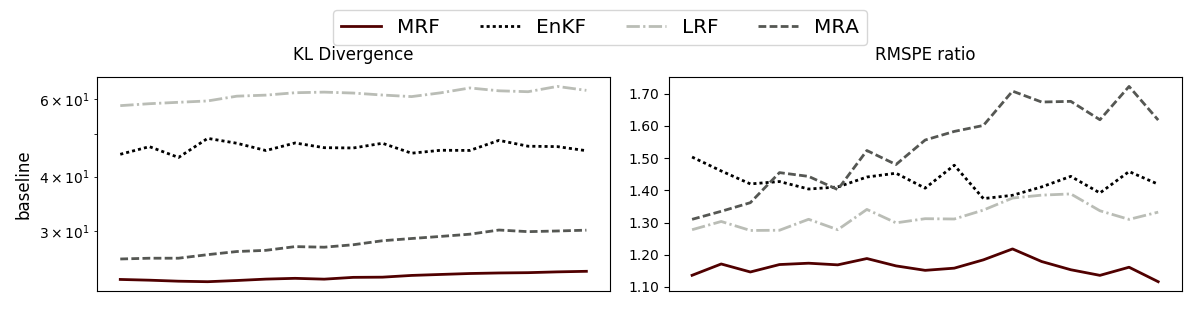}
	\end{subfigure}
	\begin{subfigure}{1.0\textwidth}
		\includegraphics[trim={0, 3mm, 0, 5mm}, clip, width=1.0\textwidth]{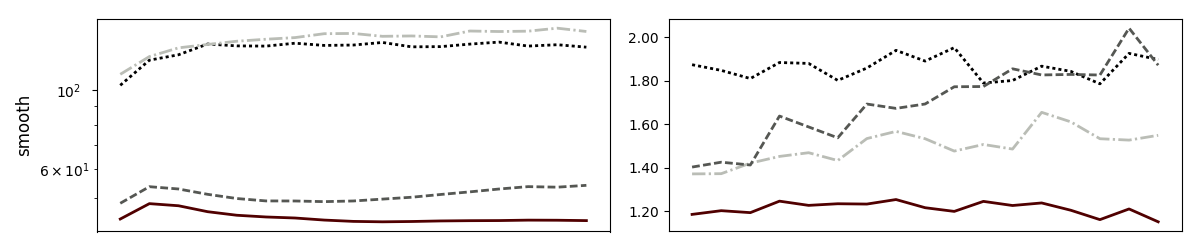}
	\end{subfigure}
	\begin{subfigure}{1.0\textwidth}
		\includegraphics[trim={0, 3mm, 0, 5mm}, clip, width=1.0\textwidth]{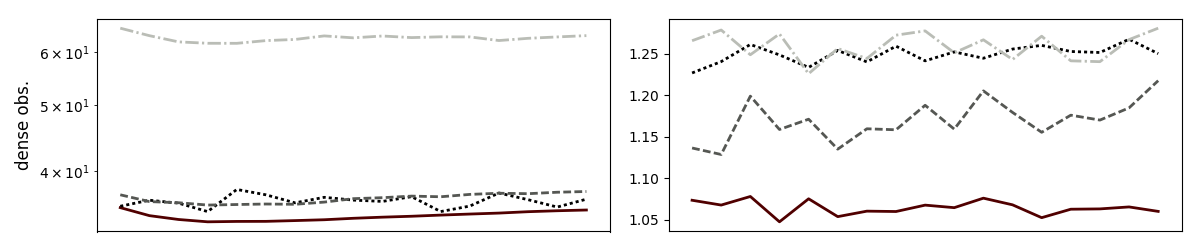}
	\end{subfigure}
	\begin{subfigure}{1.0\textwidth}
		\includegraphics[trim={0, 4mm, 0, 5mm}, clip, width=1.0\textwidth]{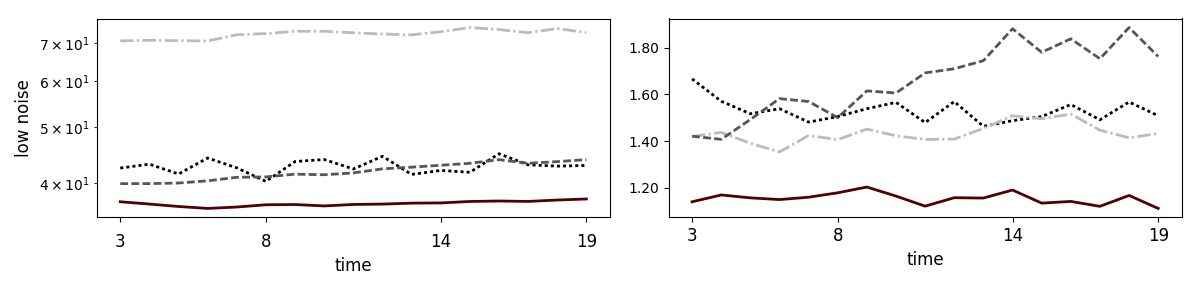}
	\end{subfigure}
	\vspace{-1mm}
	\captionof{figure}{Filter scores for different parameter settings; \textbf{one-dimensional} domain. Note that we used different scales on the vertical axis for each plot, with a logarithmic scale for the KL divergence.}
	\label{fig:scores1d}
\end{figure}

%%%%%%%%%%%%%%

\subsection{Two-dimensional domain \label{sec:simul-2d}}

We also considered a diffusion-advection model on a unit square, and we discretized it on a regular grid of size $n_\grid = 34 \times 34 = 1{,}156$. As before, we used $T=20$ evenly spaced time points. Writing the model in the linear form \eqref{eq:obs}-\eqref{eq:evol}, $\levol_t$ was a sparse matrix with nonzero entries corresponding to interactions between neighboring grid points to the right, left, top and bottom. A detailed description of the simulation, including examples of process realizations, is given in Section S4.3.

\begin{table}
	\centering
	\small
	\begin{tabular}{c|ccccc} 
	& $n_t/n_\grid$ & $\nu$ & $\lambda$ & $\sigma^2_w$ & $\sigma^2_v$ \\
	\hline \hline
	baseline & 0.1 & 0.5 & 0.15 & 0.5 & 0.25 \\
	\hline
	smooth & 0.1 & \textbf{1.5} & 0.15 & 0.5 & 0.25 \\
	dense obs.\ & \textbf{0.3} & 0.5 & 0.15 & 0.5 & 0.25 \\
	low noise & 0.1 & 0.5 & 0.15 & 0.5 & \textbf{0.1}
	\end{tabular}
	 \caption{Settings used in the \textbf{2D} simulation. Bold values indicate changes with respect to the baseline.}
	 \label{tab:baseline-2d}
\end{table}

Similar to the 1D case, we first considered baseline parameter settings and then we changed some of them, one at a time. The multi-resolution decomposition used $M=4$ and, similar to \citet{Katzfuss2016mra} we changed $J_m$ across resolutions $m$: $(J_1, \dots, J_4) = (2, 4, 4, 4)$. We also varied the numbers of knots $r_m$ used at each resolution: $(r_0, \dots, r_4) = (16, 8, 6, 6, 6)$. Thus, to achieve a fair comparison, we used $N=42$ for EnKF, LRF, and MRA. As shown in Figure \ref{fig:scores-2d}, MRF again performed best in all four scenarios.

\begin{figure}
	\begin{subfigure}{1.0\textwidth}
		\includegraphics[trim={0, 3mm, 0, 0}, clip, width=1.0\textwidth]{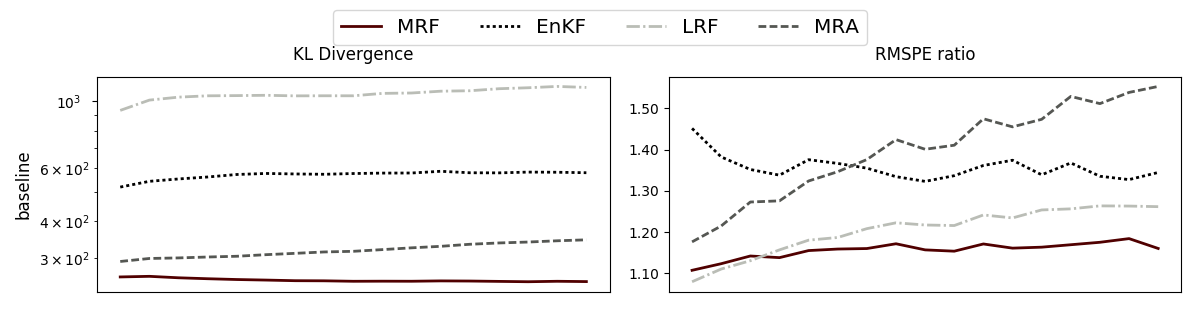}
	\end{subfigure}
	\begin{subfigure}{1.0\textwidth}
		\includegraphics[trim={0, 3mm, 0, 5mm}, clip, width=1.0\textwidth]{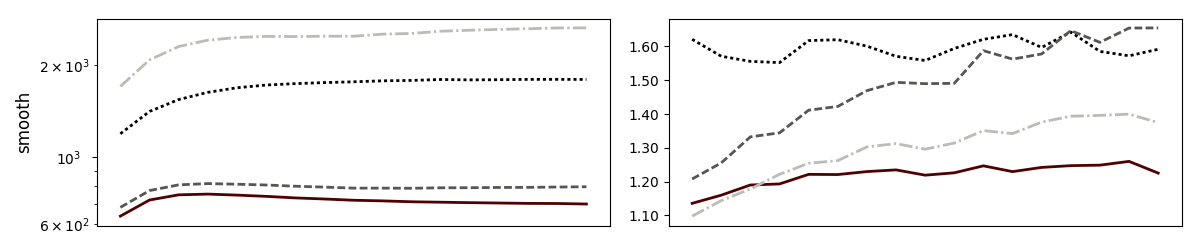}
	\end{subfigure}
	\begin{subfigure}{1.0\textwidth}
		\includegraphics[trim={0, 3mm, 0, 5mm}, clip, width=1.0\textwidth]{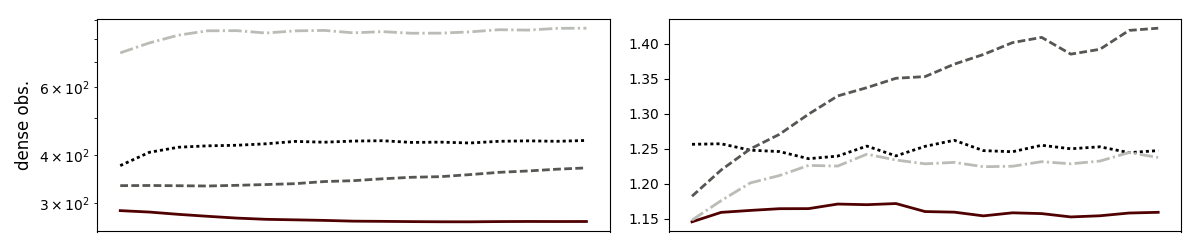}
	\end{subfigure}
	\begin{subfigure}{1.0\textwidth}
		\includegraphics[trim={0, 4mm, 0, 5mm}, clip, width=1.0\textwidth]{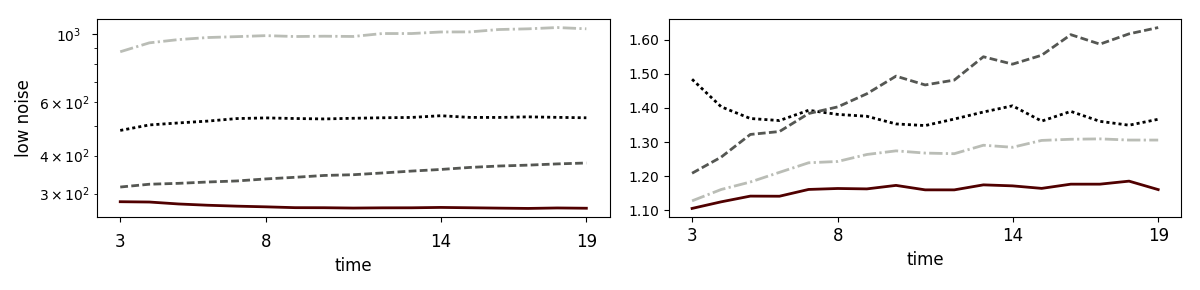}
	\end{subfigure}
	\vspace{-1mm}
	\caption{Filter scores for different parameter settings; \textbf{two-dimensional} domain. Note that we used different scales on the vertical axis for each plot, with a logarithmic scale for the KL divergence.}
	\label{fig:scores-2d}
\end{figure}

%%%%%%%%%%%%%%%%%%%%%%%%%%%%%%%%%%%%%%%%%%%%%%%%%%%%%%%

\section{Sediment movements in Lake Michigan\label{sec:data-application}}

We also considered filtering inference on sediment concentration in Lake Michigan over a period of one month, March 1998, based on satellite data. Such inference can be used by hydrologists to increase their understanding of sediment transport mechanisms and fine-tune existing domain-specific models. We closely followed an earlier study of this problem done by \citet{Stroud2010} in the context of spatio-temporal smoothing. Unless specified otherwise, we used the same model and parameter estimates. We briefly summarize the general framework below and indicate the few modifications we introduced. 

The lake area was divided into $n_\grid=14{,}558$ grid cells of size 2km $\times$ 2km each. We use $\bx_t$ to denote the sediment concentrations at the $n_\grid$ cells at time $t$. The time dimension was discretized into 409 intervals. The sediment transport model was assumed to be
$
\bx_t = \levol_t \bx_{t-1} + \bfrho_t + \bw_t,
$
where $\levol_t$ describes the temporal evolution based on a hydrological PDE model, $\bfrho_t$ is a vector with external inputs representing the influence of water velocity and bottom sheer stress, and the model error $\bw_t$ is assumed  to follow a $\normal(\bfzero, \bQ_t)$ distribution with covariance matrix $\bQ_t = (\sigma_\omega^2 \bfOmega_t \bfOmega'_t) \circ  \bT$, where $\circ$ denotes element-wise multiplication. All matrices $\bfOmega_t$ have dimensions $n_\grid \times 5$ and reflect the spatial structure of the error in the original study, while $\bT$ is taken to be a tapering matrix based on a Kanter covariance function with a tapering radius that leaves about 200 nonzero elements in each row.

The data comprise 10 satellite measurements of remote-sensing reflectance (RSR) at the frequency of 555 nm taken over the southern basin of Lake Michigan, modified in a way that accounts for the effects of the cloud cover. The observed value at each grid point was assumed to be the first-order Taylor expansion of $h(c) = \theta_0 + \theta_1\log(1+ \theta_2(c + \theta_3))$ taken around the initial mean of the sediment concentration at time $t=0$. 
Using $\by_t$ to denote the vector of observations at time $t$ after removing a time-varying instrument bias and accounting for constant terms in the Taylor expansion, we assumed $\by_t = \bH_t\bx_t + \bv_t$ as in \eqref{eq:obs},
where $\bH_t$ had only one nonzero element in each row, $\bv_t \sim \normal(\bfzero, \bR_t)$, and $\bR_t$ was diagonal. 

Because of the moderate size of the spatial grid, we were able to compute the exact Kalman filter solution. We set $M=5$, $J=4$, and $(r_0, \dots, r_5) = (16, 8, 8, 8, 4, 4)$ for the MRF, which implied that $N=\sum_m r_m = 48$ for the other approximation methods in Section \ref{sec:simulations}. The tapering range used in EnKF was selected such that the tapering matrix had only 5 nonzero elements per row, which corresponds to the setting used by \citet{Stroud2010}. While this is inconsistent with the comparison principles outlined in Section \ref{sec:simulations}, it made the EnKF perform better in this case. 

\begin{table}
	\caption{Root average squared difference (RASD) between approximate and exact filtering means for sediment concentration}
	\centering
	\small
	\begin{tabular}{c|cccc}
		 & MRF & EnKF & LRF & MRA \\
		\hline
		RASD & 0.08 & 0.22 & 0.42 & 0.72
	\end{tabular}

	\label{tab:mse}
\end{table}

As the true concentrations were unknown, we compared the approximate filtering means to the exact means obtained by the Kalman filter, using the root average squared difference $\big(\sum_t \sum_i (\hat\bfmu_t[i] - \bfmu^{KF}_t[i])^2\big)^{1/2}$ between the approximate filtering means $\hat\bfmu_t[i]$ and the KF means $\bfmu^{KF}_t[i]$, averaged over all times $t$ and grid points $i$. The results, reported in Table \ref{tab:mse}, show the MRF outperforming all other approximate methods. To visually verify these results, we also present satellite data and sediment concentration estimates for three selected time points in Figure \ref{fig:preds}. A video with all time points can be found at \url{http://spatial.stat.tamu.edu}.

\begin{figure}
	\centering
	\begin{tabular}{ccc}
		\includegraphics[angle=90,origin=c,width=0.31\textwidth]{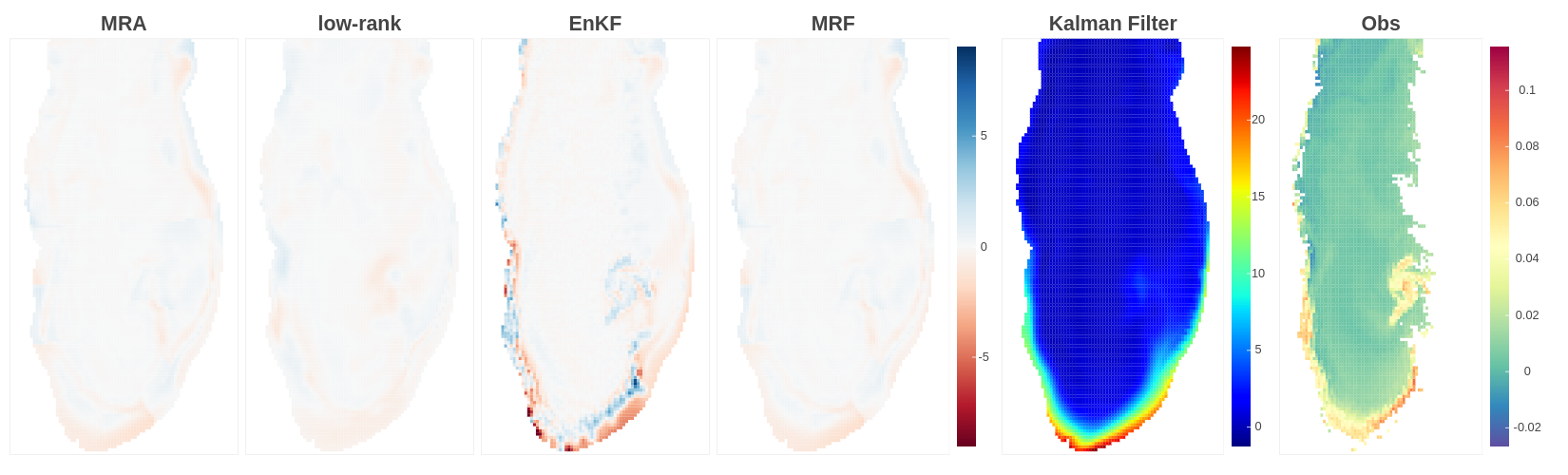} &
		 \includegraphics[angle=90,origin=c,width=0.31\textwidth]{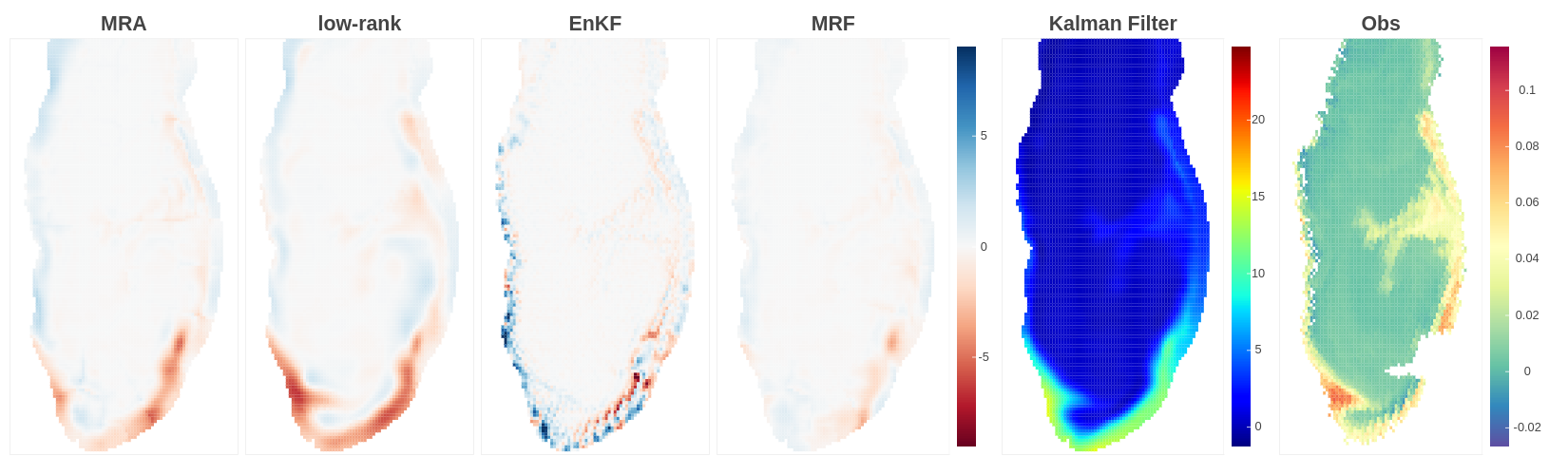} &
		  \includegraphics[angle=90,origin=c,width=0.31\textwidth]{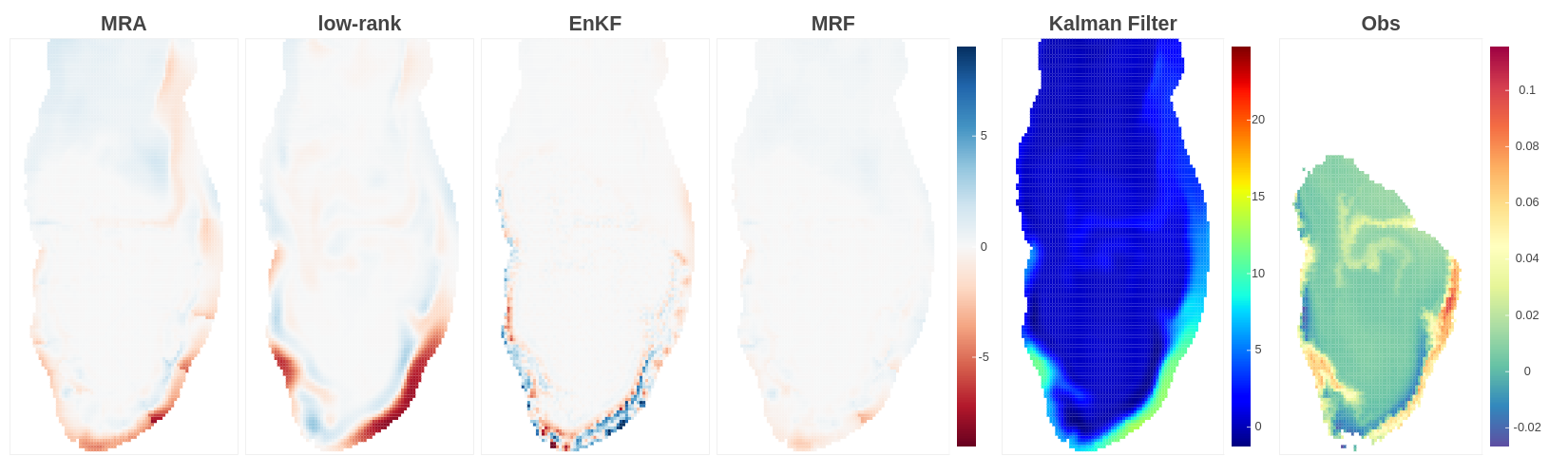} \\
		&&\\
	\small	(a) t = 2 & (b) t= 267 & (c) t=409
	\vspace{1mm}	
	\end{tabular}
	\caption{Satellite data (in log RSR) and exact Kalman filtering means of sediment concentrations (in mg/L), along with differences of approximate filtering means to the Kalman filter. We display the results for the southern basin of the lake only, where differences between the methods are most pronounced.}
	\label{fig:preds}
\end{figure}

For a grid of the size $n_\grid$ considered here, a single step of the MRF took roughly 5\% of the time required by the exact Kalman filter (on a laptop with 8GB of memory and Intel(R) Core(TM) i7-3630QM CPU @ 2.40GHz). More importantly, the MRF scales well to even larger grid sizes (see Section \ref{sec:comp-cost}), while exact calculations will quickly become infeasible due to memory constraints. Exact computation times and memory limitations will, of course, depend heavily on the computational environment.

%%%%%%%%%%%%%%%%%%%%%%%%%%%%%%%%%%%%%%%%%%%%%%%%%%%%%%%
\section{Conclusions and future work}\label{sec:conclusion}

We introduced the multi-resolution filter (MRF), a new filtering method for linear Gaussian spatio-temporal state-space models, which relies on a block-sparse multi-resolution matrix decomposition. We proved that the sparsity can be preserved under filtering through time, ensuring scalability of the MRF to very large spatial grids.
In our comparisons, the MRF substantially outperformed existing methods that can be used to approximate the Kalman filter. We also successfully applied the MRF to inferring sediment concentration in Lake Michigan.

% approx better than simplified model
Spatio-temporal processes typically exhibit highly complicated structures that make exact inference intractable, especially in high dimensions. We believe that it is often better to conduct approximate inference for a realistic, intractable model, rather than carrying out ``exact'' inference for a crude simplification (e.g., a low-rank version) of the model. While it might be challenging to precisely quantify the approximation accuracy in the former case (e.g., for the MRF), approximate inference can give better results than exact inference in a simplified model, which often completely ignores the error incurred by simplifying the model.

While we have focused on spatio-temporal data here, our methods are also applicable to general SSMs of the form \eqref{eq:obs}--\eqref{eq:evol} that do not correspond to physical space and time, as long as some distance between the elements of each state vector can be specified. 

Potential future work includes extensions to non-Gaussian data, nonlinear evolution, and smoothing inference. We are also developing a user-friendly implementation of the MRF with sensible default settings for the number of knots and domain partitioning.

%%%%%%%%%%%%%%%%%%%%%%%%%%%%%%%%%%%%%%%%%%%%%%%%%%%%%%%
\footnotesize
\appendix
\section*{Acknowledgments}

The authors were partially supported by National Science Foundation (NSF) Grants DMS--1521676 and DMS--1654083. We would like to thank Jacob Bien, Wenlong Gong, Joseph Guinness, Dorit Hammerling, Ephrahim Hanks, Peter Kuchement, Mohsen Pourahmadi, Ramalingam Saravanan, Michael Stein, Jonathan Stroud, Istvan Szunyogh, Christopher Wikle, Catherine Yan, and several reviewers for helpful comments and discussions. Special thanks to Jonathan Stroud for providing code and data for the Lake Michigan example.

%%%%%%%%%%%%%%%%%%%%%%%%%%%%%%%%%%%%%%%%%%%%%%%%%%%%%%%

\section{Proofs \label{app:proofs}}

We now provide proofs for the propositions stated throughout the article. We simplify notation by dropping most time subscripts; to avoid confusion, we denote $\bB_{t|t-1}$ by $\bB$, and $\bB_{t|t}$ by $\tilde{\bB}$. In Section S8, we provide lemmas with proofs that are used in the proof of Proposition \ref{prop:pd} here. Sections S6-S7 contain additional technical concepts used in the lemmas, including a review of basic ideas from graph theory, hierarchical-matrix theory, and some illustrative figures. Finally, throughout this appendix, if $\bG$ is a square matrix, we use $\bG^L$ and $\bG^U$ to denote its lower and upper triangles, respectively.

%%%%%%%
\begin{proof}[Proof of Proposition \ref{prop:sparsity-B}]
Recall that $\bB = \left( \bB^M, \bB^{M-1}, \ldots, \bB^0 \right)$. 
This lets us write the $i$-th row of $\bB$ as $\bB[i,:] = \left( \bB^M[i,:], \bB^{M-1}[i,:], \ldots, \bB^0[i,:] \right)$.
By construction, each block $\bB^m$ is block-diagonal and such that for $m\leq M$, each segment $\bB^m[i,:]$ has only $r_m$ nonzero elements.
Because each row of $\bB$ is composed of $M+1$ blocks $\bB^m[i,:]$ for $m=0, \dots, M$, this ends the proof.
\end{proof}

%%%%%%%
\begin{proof}[Proof of Proposition \ref{prop:blockBB}]
Direct calculation shows that $\bB'\bB$ is a block matrix consisting of $(M+1) \times (M+1)$ blocks with $(M-k+1,M-l+1)$-th block $(\bB^{k})'\bB^{l}$. 
Since for each $j$ the matrix $\bB^j$ has dimensions $n_\grid \times |\knots^j|$ it follows that $(\bB^k)'\bB^l$ is of size $|\knots^k| \times |\knots^l|$. 
Note that $\bB^k$ and $\bB^l$ are block-diagonal with blocks of size $|\gridind_\jk| \times r_k$ and $|\gridind_{j_1, \dots j_l}| \times r_l$, respectively. 
Assuming without loss of generality that $k\leq l$, we have that
\begin{equation} \label{eq:block-inclusion}
{\textstyle \gridind_\jk = \bigcup_{j_{k+1}=1}^J \dots \bigcup_{j_l=1}^J \gridind_\jl, \quad\quad\quad \implies \quad\quad\quad |\gridind_\jk| = \sum_{j_{k+1}=1}^J \dots \sum_{j_l=1}^J |\gridind_\jl|}.
\end{equation}

Thus $\bB^l$ can be viewed as a block-diagonal matrix with blocks of height $|\gridind_\jk|$. 
We can also determine their width to be $w_\jk = \sum_{j_{k+1}=1}^J \dots \sum_{j_l=1}^J |\knots_{j_1, \dots, j_k, j_{k+1}, \dots j_l}|$.
This means $(\bB^k)'\bB^l$ is the product of two block-diagonal matrices with matching block sizes. 
Therefore the product will be also block-diagonal with blocks of dimensions $w_\jk\times r_k$.
\end{proof}

 %%%%%%%
\begin{proof}[Proof of Proposition \ref{prop:pd}] 
~
\begin{enumerate}	
\item Observe that under Assumption \ref{ass:hr}, $\bR^{-1}$ and $\bH$ are block-diagonal with blocks of matching dimensions. Since $\bR^{-1}$ has square blocks, we conclude that $\bH'\bR^{-1} \in \Sp(\bH')$. Thus, if $\tilde{\bR}^{-1} \colonequals \bH'\bR^{-1}\bH$, then $\tilde{\bR}^{-1} \in \Sp(\bH'\bH)$. The latter is a block-diagonal matrix with square blocks of size $|\gridind_\jM|$.

Next, we demonstrate that $\bB'\tilde{\bR}^{-1}\bB \in \Sp(\bB'\bB)$. First, as $\tilde{\bR}^{-1}$ is block-diagonal, the $(M+1-k,M+1-l)$-th block of $\bB'\tilde{\bR}^{-1}\bB$ is given by $(\bB^{k})' \tilde{\bR}^{-1} \bB^{l}$.

Now, for each $0\leq k \leq M$, $\bB^k$ is a block-diagonal matrix with blocks of size $|\gridind_\jk| \times r_k$, but $\tilde{\bR}^{-1}$ has blocks of size $|\gridind_\jM| \times |\gridind_\jM|$. However, recalling (\ref{eq:block-inclusion}), blocks of $\tilde{\bR}^{-1}$ can also be viewed as having dimensions $|\gridind_\jk| \times r_k$. Because this implies that $(\bB^k)'\tilde{\bR}^{-1} \in \Sp(\bB^k)'$, we have $(\bB^k)'\tilde{\bR}^{-1}\bB^l \in \Sp((\bB^k)' \bB^l)$ and hence $(\bB)'\tilde{\bR}^{-1}\bB \in \Sp (\bB'\bB)$. Finally, we conclude that $\prec \in \Sp(\bB'\bB)$, because $\prec=\bI_{n_\grid} + \bB'\tilde{\bR}^{-1}\bB$ and all diagonal elements of $\bB'\bB$ are nonzero.
	
	\item According to \citet[][Thm.~1]{khare2012sparse}, for any positive definite matrix $\bS$, the sparsity pattern in the Cholesky factor and its inverse are the same as that of the lower triangle of $\bS$, if
(a) the pattern of zeros in $\bS$ corresponds to a homogeneous graph, and  
(b) the order of the vertices of the graph implied by the order of the rows is a Hasse-tree-based elimination scheme.
Lemmas S1 and S2 in Section S8 show that these two conditions are met for $\bB'\bB$. These lemmas, together with Part 1 above, imply that $\bL \in \Sp(\prec^L)$ and $\bL^{-1} \in \Sp(\prec^L)$.

	 \item Observe that $\prec^{-1} = (\bL\bL')^{-1} = (\bL^{-1})' \bL^{-1}$. Thus $(\bL^{-1})'$ is the Cholesky factor of $\prec^{-1}$. Moreover, by Part 2, $(\bL^{-1})' \in \Sp((\bB'\bB)^U)$. This allows us to define blocks $\blt^{m,k}$ such that
	 \begin{equation*}
	 (\bL^{-1})' = \left[ \begin{array}{cccc}
		 \blt^{M,M}  	& \ldots 		& \blt^{M,1}	& \blt^{M,0} \\
		 \vdots 	   	& \ddots 		& \vdots 		& \vdots \\
		 \bfzero 		& \ldots 		& \blt^{1,1}	& \blt^{1,0} \\
		 \bfzero 		& \ldots 		& \bfzero 		& \blt^{0,0}
		 \end{array} \right] = \left[ \blt^{m,k} \right]_{m,k = M, \ldots, 0},
	 \end{equation*}
	 where each $\blt^{m,k} \in \Sp((\bB^m)'\bB^k)$ for $m \geq k$ and is zero when $m<k$. This means that for each $m,k$ with $m \geq k$, we can consider the sparsity of $\bB^m(\bB^m)'\bB^k$ instead of $\bB^m \blt^{m,k}$.

Recall that $\bB^k$ is block-diagonal with blocks of size $|\gridind_\jk|\times r_k$. Similarly, $\bB^m$ has blocks that are $|\gridind_\jm|\times r_m$. 
However, since $k\leq m$, using (\ref{eq:block-inclusion})

we can also see $\bB^m$ as a block-diagonal matrix whose blocks have dimensions $|\gridind_\jk|\times r_k$ (cf.\ proof of Proposition \ref{prop:blockBB}). 
This implies that $\bB^m(\bB^m)' \in \Sp(\bB^k(\bB^k)')$, which means that $\bB^m(\bB^m)'\bB^k \in \Sp(\bB^k)$ and hence $\bB^m\blt^{m,k} \in \Sp(\bB^k)$.
	
Finally, we observe that
	\begin{equation*}
	\bB \cdot (\bL^{-1})' = \left[ \bB^M \; \bB^{M-1} \; \ldots \; \bB^0 \right] 
	 \cdot \left[ \begin{array}{cccc}
		 \blt^{M,M}  	& \ldots 		& \blt^{M,1}	& \blt^{M,0} \\
		 \vdots 	   	& \ddots 		& \vdots 		& \vdots \\
		 \bfzero 		& \ldots 		& \blt^{1,1}	& \blt^{1,0} \\
		 \bfzero 		& \ldots 		& \bfzero 		& \blt^{0,0}
		 \end{array} \right] = 
		 \left[ \tilde{\bB}^M \; \tilde{\bB}^{M-1} \; \ldots \; \tilde{\bB}^0 \right],
	\end{equation*}
	where $\tilde{\bB}^k = \sum_{m=k} ^M \bB^m \blt^{m,k}$. Since we showed that $\bB^m\blt^{m,k} \in \Sp(\bB^k)$, this means that $\tilde{\bB}^k \in \Sp(\bB^k)$.
\end{enumerate}
\end{proof}

%%%%%%%
\begin{proof}[Proof of Proposition \ref{prop:sparsity}]
By Proposition \ref{prop:pd}, Parts 1 and 2, $\bL \in \Sp((\bB'\bB)^L)$. 
Therefore, it suffices to show that $\bc_j=(\bB'\bB)^L[:,j]$, the $j$-th column of $(\bB'\bB)^L$, has $\mathcal{O}(N)$ nonzero elements for each $j$. 
Notice that $\bc_j = (0, \dots, 0, \bc^{k,k}_j, \dots \bc^{M,k}_j)'$ where $\bc^{k,l}_j = ((\bB^k)'\bB^l)[:,j]$, the $j$-th column of $(\bB^k)'\bB^l$. 
Because $l\geq k$, each of the $\bB^k(\bB^l)'$ matrices is block-diagonal with blocks of height $|\knots_\jk|$. 
The vector $\bc_j^{k,l}$ intersects exactly one of such diagonal blocks, and so the total number of nonzero elements in $\bc$ is at most $N=\sum_m r_m$.
\end{proof}

%%%%%%%
\begin{proof}[Proof of Proposition \ref{prop:mrd-comp}]
Observe that it is enough to consider only the complexity of operations in \eqref{eq:W} because $\bV^l_\jm$ can be obtained by selecting appropriate rows from $\bW_\jm^l$. 
Given matrix $\bfSigma$, we only need to calculate the second term in $\eqref{eq:W}$. First, note that calculating $\bW^l_\jm$ for all $(\jm)$ is the same as computing $\bW^l_\jl$ for all $l$, and then, for each $(\jm)$, selecting the rows corresponding to $\gridind_\jm$. Thus we show the complexity of calculating all $\bW^l_\jl$. 

Assume that all $\bW^k_\jl$ for $k<l$ are already given and consider the summation term.
Each of its components takes $\mathcal{O}(|\gridind_\jl|r^2 + r^3 + r^3 + r^3) = \mathcal{O}(|\gridind_\jl|r^2)$ to compute. 
Because for any given $l$, there are at most $M$ terms under the summation, their joint computation time is $\mathcal{O}(M\cdot|\gridind_\jl|r^2)$.
For a given $l$, these calculations have to be performed for each set of indices $\gridind_\jl$.
Thus, obtaining all $\bW^l_\jl$ requires $\mathcal{O}(M\cdot\sum_{\jl}|\gridind_\jl|r^2) = \mathcal{O}(M\cdot n r^2)$ time.
Now notice that $\gridind_\jm \subset \gridind_\jl$.
Therefore, once we have $\bW^l_\jl$, we obtain $\bW^l_\jm$ by selecting appropriate rows from $\bW^l_\jl$. Finally, iterating over $l=0, \dots, M$ means that the total cost of Algorithm 2 is $\mathcal{O}(M^2 n r^2) = \mathcal{O}(nN^2)$.
\end{proof}

%%%%%%%
\begin{proof}[Proof of Proposition \ref{prop:mrf-complex}]

The forecast step requires calculating $\bfmu_{t|t-1} = \levol_t\bfmu_{t-1|t-1}$ and $\bB_{t|t-1}^F = \levol_t\bB_{t-1|t-1}$, 
which can be obtained in $\order(nr)$ and $\order(nrN)$ time, respectively, due to the sparsity structures of $\bB_{t-1|t-1}$ (see Proposition \ref{prop:sparsity-B}) and $\levol_t$ (Assumption \ref{ass:local}).

By Proposition \ref{prop:mrd-comp}, the MRD of a given covariance matrix $\bfSigma$ requires $\mathcal{O}(nN^2)$ operations. Here, $\bfSigma=\bfSigma_{t|t-1}$ is not given, but each $(i,j)$ element must be computed as
\[
\bfSigma_{t|t-1}[i,j] = (\bB_{t|t-1}^F[i,\,:\,])(\bB_{t|t-1}^F[j,\,:\,])'+ \bQ_t[i,j].
\]
This does not increase the complexity of the MRD, because the MRD requires only $\order(nN)$ elements of $\bfSigma_{t|t-1}$, each of which can be computed in $\order(N)$ time due to the sparsity structure of $\bB^F_{t|t-1}$. 
Thus, the entire forecast step can be performed in $\mathcal{O}(nN^2)$ time.

In the update step, we must compute $\pprec$, $\bL^{-1}=\pprec^{-1/2}$, and $\bB_{t|t}=\bB_{t|t-1} (\bL^{-1})'$.
Under Assumption \ref{ass:hr}, $\bH$ and $\bR$ are block-diagonal matrices with at most $J^M$ blocks of size $\order(r \times r)$ each. 
Thus, calculating $\tilde{\bR} \colonequals \bH'\bR^{-1}\bH$ requires $\order(J^M r^3)=\order(n r^2)$ operations. 
The resulting matrix is block-diagonal with blocks of size $\order(r \times r)$, conformable with the blocks of $\bB_{t|t-1}$. 
Given $\tilde{\bR}$, the cost of calculating $\pprec$ is dominated by multiplying $\bB_{t|t-1}$ by $\tilde{\bR}$. 
By Proposition \ref{prop:sparsity-B}, each row of $\bB_{t|t-1}$ has $N$ nonzero elements, so in view of the structure of $\tilde{\bR}$ determined above, 
it takes $\mathcal{O}(nN^2)$ operations to obtain the product $\bB_{t|t-1}\tilde{\bR}$ and, consequently, to calculate $\pprec$.

The complexity of computing a Cholesky factor is on the order of the sum of the squared number of nonzero elements per column \citep[e.g.,][Thm.~2.2]{Toledo2007}. 
Thus, computing $\bL$ requires $\order(n N^2)$ time, because $\bL$ has $\order(N)$ elements in each of its $n$ columns (Proposition \ref{prop:pd}).
Computing $\bL^{-1}$ can be accomplished by solving a triangular system of equations for each column of $\bL^{-1}$. 
Using Proposition \ref{prop:sparsity}, we conclude that each of these systems will have only $\order(N)$ equations and thus can be solved in $\order(N^2)$ time \citep[Ch.~4.2]{kincaid2002numerical}. 
As we need to compute $n$ columns, the total effort required for obtaining $\bL^{-1}$ is $\order(n N^2)$.

Finally, recall that both $\bB_{t|t-1}$ and $\bL^{-1}$ have $\order(N)$ elements in each row and that, by Proposition \ref{prop:pd}, their product, $\bB_{t|t}$, has only $\order(nN)$ nonzero elements. Because each of these elements can be computed in $\order(N)$ time, the total computation cost of this step is $\order(nN^2)$.

To summarize, all three matrices necessary in the update step can be obtained in $\mathcal{O}(nN^2)$ time. Thus, we showed that both steps of Algorithm \ref{alg:mrf} require $\mathcal{O}(nN^2)$ time, which completes the proof.
\end{proof}

%%%%%%%
\begin{proof}[Proof of Proposition \ref{prop:hodlr}]
For $m=1,\ldots,M$, define $\bB^{0:m} = (\bB^m, \dots, \bB^0)$ as the submatrix of $\bB$ consisting of the column blocks corresponding to resolutions $0,\dots,m$. To show that $\bB\bB' \in H_M^r$, we prove by induction over $m=1,\ldots,M$ that $(\bB^{0:m}\bB^{0:m})' \in H_m^r$. For $m=1$, we have $\bB^{0:1} = \left[ \begin{array}{ccc} \bB_{1} & \bfzero & \bB_{01}\\ \bfzero & \bB_{2} & \bB_{02}\end{array} \right]$, where $\bB_{01}$ and $\bB_{02}$ are each $r$ columns wide. Thus, 
\[
\bB^{0:1}(\bB^{0:1})' = \left[ \begin{array}{cc}\bB_{01} \bB_{01}' + \bB_1 \bB_1' & \bB_{01} \bB_{02}' \\ \bB_{02} \bB_{01}' & \bB_{02} \bB_{02}' + \bB_2 \bB_2' \end{array} \right].
\]
and so $\bB^{0:1}(\bB^{0:1})' \in H_1^r$. 

Now, assume that $\bB^{0:m-1} (\bB^{0:m-1})' \in H_{m-1}^r$. We have
\[
\textstyle \bB^{0:m} (\bB^{0:m})' = \sum _{j=0} ^{m} \bB^j (\bB^j)' = \sum _{j=0} ^{m-1} \bB^j (\bB^j)' + \bB^m (\bB^m)' = \bB^{0:m-1} (\bB^{0:m-1})'  + \bB^m (\bB^m)'.
\]
Next observe that for any $k$, the matrix $\bB^k$ is block-diagonal, which means that $\bB^k (\bB^k)'$ is also block-diagonal with dense blocks $\bB^k(\bB^k)'[\gridind_{\jk},\gridind_{\jk}]$. However, recursive partitioning of the domain means that $\gridind_{j_1, \dots, j_{k-1}} \supset \gridind_{\jk}$. Therefore, if $k>j$, then blocks of $\bB^k (\bB^k)'$ are nested within the blocks of $\bB^j (\bB^j)'$. Since this holds also for $k=m-1$ and $j=m$, it means that $\bB^{1:m} (\bB^{1:m})' \in H_{m}^r$.
\end{proof}

%%%%%%%%%%%%%%
\bibliographystyle{apalike}
\bibliography{lib}

\end{document}